\documentclass[journal,draftclsnofoot,onecolumn,12pt,twoside]{IEEEtran}
\usepackage{amsmath,amsfonts}
\usepackage{algorithmic}
\usepackage{algorithm}
\usepackage{amssymb}

\usepackage{color}
\usepackage{float} 

\usepackage{array}
\usepackage[caption=false,font=normalsize,labelfont=sf,textfont=sf]{subfig}
\usepackage{textcomp}
\usepackage{stfloats}
\usepackage{url}
\usepackage{bm}
\usepackage{setspace}
\usepackage{verbatim}
\usepackage{multirow} 
\usepackage{ntheorem}
\usepackage{graphicx}
\usepackage{cite}
\allowdisplaybreaks[4]
\theorembodyfont{\upshape}
\theoremheaderfont{\it}
\qedsymbol{\square}
\theoremseparator{:}
\newtheorem{theorem}{\indent Theorem}
\newtheorem{lemma}{\indent Lemma}
\newtheorem*{proof}{\indent Proof}
\newtheorem{remark}{\indent Remark}
\newtheorem{corollary}{\indent Corollary}

\makeatletter

\newcommand{\Rmnum}[1]{\expandafter\@slowromancap\romannumeral #1@}
\makeatother
\begin{document}

\title{Superimposed RIS-phase Modulation for MIMO Communications: A Novel Paradigm of Information Transfer}

\author{Jiacheng~Yao,
		Jindan~Xu, \IEEEmembership{Member,~IEEE,}
		Wei~Xu,~\IEEEmembership{Senior~Member,~IEEE,}
		Chau~Yuen,~\IEEEmembership{Fellow,~IEEE,}
        and~Xiaohu~You,~\IEEEmembership{Fellow,~IEEE}
\thanks{J. Yao, W. Xu, and X. You are with the National Mobile Communications Research Laboratory (NCRL), Southeast University, Nanjing 210096, China (\{jcyao, wxu, xhyu\}@seu.edu.cn).}
\thanks{J. Xu and C. Yuen are with the School of Electrical and Electronics Engineering, Nanyang Technological University, Singapore 639798, Singapore (e-mail: jindan1025@gmail.com, chau.yuen@ntu.edu.sg).}}

%

\maketitle
\vspace{-1.95cm}
\begin{abstract}
Reconfigurable intelligent surface (RIS) is regarded as an important enabling technology for the sixth-generation (6G) network. Recently, modulating information in reflection patterns of RIS, referred to as reflection modulation (RM), has been proven in theory to have the potential of achieving higher transmission rate than existing passive beamforming (PBF) schemes of RIS. To fully unlock this potential of RM, we propose a novel superimposed RIS-phase modulation (SRPM) scheme for multiple-input multiple-output (MIMO) systems, where tunable phase offsets are superimposed onto predetermined RIS phases to bear extra information messages. The proposed SRPM establishes a universal framework for RM, which retrieves various existing RM-based schemes as special cases. Moreover, the advantages and applicability of the SRPM in practice is also validated in theory by analytical characterization of its performance in terms of average bit error rate (ABER) and ergodic capacity. To maximize the performance gain, we formulate a general precoding optimization at the base station (BS) for a single-stream case with uncorrelated channels and obtain the optimal SRPM design via the semidefinite relaxation (SDR) technique. Furthermore, to avoid extremely high complexity in maximum likelihood (ML) detection for the SRPM, we propose a sphere decoding (SD)-based layered detection method with near-ML performance and much lower complexity. Numerical results demonstrate the effectiveness of SRPM, precoding optimization, and detection design. It is verified that the proposed SRPM achieves a higher diversity order than that of existing RM-based schemes and outperforms PBF significantly especially when the transmitter is equipped with limited radio-frequency (RF) chains. 
\end{abstract}
\vspace{-0.5cm}

\section{Introduction}
In order to support emerging wireless applications, such as fully immersive eXtended Reality (XR), the requirement of an order-of-magnitude higher rate than that of the current wireless network is one of the key performance indicators (KPI) of the sixth-generation (6G) network \cite{you,6g2}. In addition, the ability of the network to collect sensory data from environment is the driving force to build intelligence for 6G \cite{zongshu,isac1,gomore}, which also requires massive data transmission and extremely high data rate. Evolution of massive multiple-input multiple-output (MIMO) is viewed as a promising physical-layer technique enabling extremely high data rate \cite{mimo}. By scaling up the antenna array, the spectral efficiency of massive MIMO can be greatly improved due to high multiplexing gain. However, this performance gain comes at the cost of prohibitive hardware cost and high energy consumption, which hinders its use in practice.

As a remedy, reconfigurable intelligent surface (RIS), made up of a large number of low-cost passive reflecting antenna elements, has received worldwide attention. It is regarded as an important enabling key technology for the 6G network \cite{ris1,ris2,ris3}. Specifically, each reflecting element of RIS can independently induce phase changes to the incident electromagnetic waves \cite{ris3}, thereby realizing intelligent manipulation of the electromagnetic wave propagation environment. Due to its features as a  passive and low-cost device, RIS can be densely deployed to enhance the coverage of wireless networks with significantly reduced energy consumption. In particular, RIS has been widely used for signal enhancement\cite{wu}\cite{yao}, coverage expansion \cite{cover}, physical layer security\cite{sop}, and interference suppression \cite{inter}. Through passive beamforming (PBF) at RIS, we can dynamically control the signal-to-noise ratio (SNR) at receivers \cite{pbf2,pbf3}. However, these studies exploited RIS as purely a passive reflector to control the propagation environment, while it has been revealed in \cite{theo1} that RIS can transfer extra information in order to achieve higher communication rate close to the theoretical RIS channel capacity.

To realize the extra information transfer by RIS, there are mainly two types of schemes in literature, namely, spatial modulation (SM) and reflection modulation (RM). In the concept of SM, only a part of the antennas are activated and the additional information is implicitly transmitted by the indices of these activated antennas \cite{sm}. By exploiting RIS with SM, the beam can be better focused through intelligent reflecting at RIS, which provides new possibilities for SM. For instance, RIS-space shift keying (RIS-SSK) and RIS-spatial modulation (RIS-SM) schemes were proposed in \cite{rsm} by using the receive antenna indices for information transmission.  To further improve the performance, the power allocation and phase shifts at the RIS were jointly optimized in \cite{tsm} for separate  transmit SM (TSM) and receive SM (RSM). Moreover, an integration of  using both transmit and receive antenna indices for SM was discussed in \cite{gsm}. However, the achieved data rate of SM is greatly limited by the number of antennas, which is usually not large due to hardware cost. Additionally, in these RIS-assisted SM schemes, RIS itself does not transmit any information. Its essence is still to exploit PBF to realize intelligent control of the received SNR, which is proven strictly suboptimal in utilizing the RIS and cannot fully unlock its potential \cite{theo1}.

On the other hand, RM is another way of improving the spectral efficiency, which encodes the extra information in the reflection pattern of RIS \cite{xujindan,guoshuai}. By using RM, the RIS not only encodes the source information at the transmitter, but also delivers its local data as the extra information to the receiver, which is of great significance for the 6G network. For example, when sensors are deployed at the RIS, RM can be used to transmit the sensing environmental data to the receiver without deploying additional radio-frequency (RF) components and antennas \cite{pbit}. In addition,the potential of RM-enhanced full-duplex system has also been explored in \cite{stmm}. The optimality of modulating information in RIS phases has also been theoretically proven via rigorous channel capacity analysis \cite{theo1} and degree-of-freedom (DoF) analysis \cite{theo2}. To develop effective RM methods through RIS, multiple schemes were proposed  in  recent works. In \cite{pbit}, by combining the PBF and RM, a passive beamforming and information transfer (PBIT) scheme was proposed, where the ON/OFF states of each reflecting element was utilized  to deliver an additional binary message and PBF was shaped with the activated elements for performance improvement. However, since the number of activated elements is varying across transmitting symbols, the PBIT often suffers from high outage probability. To address this issue, a novel RIS-based reflection pattern modulation (RIS-RPM) scheme was proposed in \cite{rpm}, where a fixed number of reflecting elements are activated to ensure a stable communication link. Moreover, considering that turning off some reflecting elements leads to degraded signal power at the receiver, the authors in \cite{qrm} proposed to use the I/Q phases for all activated elements instead of the ON/OFF state to enhance performance. 
Then,  a novel RIS grouping method and a joint mapping codebook index efficient selection method are designed in \cite{gqrm} to improve the spectral efficiency. In \cite{yanwen}, the frequency-hopping states of RIS elements are exploited for modulating extra information, which achieves higher reflection efficiency than the ON/OFF states.
Moreover, in \cite{pm1} and \cite{pm2}, the extra information is transmitted by superimposing a specific phase rotation on the original phase. 
In order to boost a higher communication rate, the previous work \cite{my} proposed a framework of RM for an RIS-assisted multiple-input single-output (MISO) system to enable high-order modulation for RIS. Nevertheless, these RM-based schemes aimed at improving the received SNR rather than the joint encoding gain, which can cause the performance loss in terms of both ABER and capacity. Hence, optimization for minimizing the ABER is an intuitive and optimal way to improve the system performance.

In addition, these mentioned modulation schemes still face challenges in practice. Firstly, the modulation schemes for SM and RM generally rely on the assumption of accurate channel state information (CSI) at the transmitter, which is usually hard to obtain in practice \cite{xing}. Especially for RIS as a passive device, it does not support any signal transmission, reception, and processing. Besides, given that a large number of reflecting elements are deployed at RIS, the signalling overhead is considerable, which hinders the acquisition of accurate CSI involving the RIS. Moreover, regarding the detection for RM-based schemes, an effective detection method with low complexity is still missing. Existing schemes, like in \cite{rpm} and \cite{qrm}, mainly rely on the maximum likelihood (ML) detection suffering from exponentially high complexity. 

To address these challenges, we proposed a novel paradigm of RIS-based RM design that extends our previous work \cite{my} to a general MIMO case with more practical channel model. We characterized the ABER and ergodic capacity of the proposed scheme quantitatively. In addition, a low-complexity detection method is devised to approach the performance of ML detection. The main contributions of this paper are summarized as follows.
\begin{itemize}
\item To achieve higher communication rate, we propose a novel RIS phase modulation scheme for the extra information transfer by superimposing information-bearing phase offsets to the predetermined RIS phases. It is confirmed that the proposed scheme is a general framework for RM and retrieves most existing RM schemes as special cases. Furthermore, we derive the average bit error rate (ABER) and the ergodic capacity of the SRPM system in closed form under both spatially uncorrelated and correlated fading channels. We quantitatively discover the diversity order of the ABER under various setups for the RIS-assisted MIMO system. These theoretical results unveil that the proposed scheme achieves the maximal diversity order and exhibits significant superiority over the traditional PBF from the perspective of ergodic capacity. It also shows better applicability to practical cases with CSI uncertainty at the transmitter and discrete phases shifts at the RIS. 

\item Based on the derived ABER and ergodic capacity, we are able to further optimize the precoding at the BS for the single-stream case and pursue performance optimization. Specifically, we rewrite the ABER and ergodic capacity as a function of the precoding design, and formulate the precoding optimization problem as a nonconvex but quadratically constrained quadratic programming (QCQP) problem.  We exploit the semidefinite relaxation (SDR) technique to tackle this nonconvexity without loss of the optimality because the relaxation is fortunately proven to be tight for this problem. Interestingly, for a special case, the optimal solution of the precoding vector is derived in closed form.
 
\item To avoid the extremely high computational complexity of ML detection, we propose an effective sphere decoding (SD)-based layered detection algorithm for the RM-based schemes. We recast the original detection problem into a concise form suitable for the implementation of SD by rewriting the received signal. Considering that the transmitted message includes two parts, i.e., the message transmitted by BS and that conveyed by RIS phases, we construct a layered framework to search for valid symbols following the principle of SD, in which the symbols sent by BS and the message modulated into RIS phases are determined in~sequence.

\item Extensive numerical simulations are conducted to demonstrate the superiority of the proposed scheme and verify the derived analytical observations. In particular, we show that increasing the number of symbols is preferred to achieve a higher communication rate rather than increasing the modulation order. Moreover, we find that doubling the number of RIS elements brings approximately 3 dB reduction in terms of the transmit power for the simultaneous RIS modulation and reflection. It is also found that the proposed SD-based detector achieves near-ML performance with a much lower computational complexity.

\end{itemize}

The rest of this paper is organized as follows. In Section~\Rmnum{2}, we depict the system model and give the detail design of the proposed scheme. Section \Rmnum{3} derives the analytical ABER performance and the ergodic capacity of the proposed system. In Section~\Rmnum{4}, we formulate the precoding optimization problem based on the analytical results and obtain the optimal precoding design. In Section~\Rmnum{5}, we propose a SD-based layered detector for the proposed superimposed RIS-phase modulation scheme. Simulation results and conclusion are given in Sections \Rmnum{6} and \Rmnum{7}, respectively.

\textit{Notations:} $\mathbb{C}$ denotes the complex-valued space. $\mathbb{E}\{\cdot\}$ is the expectation operation. Operator $\mathrm{Re}(\cdot)$ returns the real part of an input complex number. $(\cdot)^T$, $(\cdot)^*$, and  $(\cdot)^H$ denote the transpose, conjugate, and conjugate transpose operations, respectively. $\vert \cdot \vert$ and  $\Vert \cdot \Vert$ return the modulus of a complex number and the  Euclidean norm of vectors, respectively. ${\rm{Tr}}(\cdot)$ denotes the trace of the input matrix. Operator ${\rm{diag}}\{\cdot\}$ denotes the diagonal operation. Operator $\lfloor\cdot \rfloor$ returns the largest integer less than or equal to the argument. Operator $\otimes$ represents the Kronecker product of two matrices. $\bm{I}_M$ stands for an $M\times M$ identity matrix. $\mathcal{CN}(\bm{\mu},\bm{\Sigma})$ is the distribution of a circularly symmetric complex Gaussian random vector with mean vector $\bm{\mu}$ and covariance matrix $\bm{\Sigma}$.

\section{System Model and Principle of SRPM}

\subsection{System Model}
As depicted in Fig. \ref{p1}, an RIS-aided downlink MIMO communication system is considered, where an RIS with $N$ passive reflecting elements is adopted to enhance the communication from an $N_t$-antenna BS to an $N_r$-antenna user. The RIS is connected to a controller that is used for controlling the adjustment of phase shift of each RIS element and exchanging information with the BS \cite{zhao}. We assume that the RIS is deployed in a vicinity of BS, in which case a reliable reflection link from BS to the user, low signalling overhead between the BS and the RIS controller, and low training overhead for channel estimation are available \cite{near}.  \par
\begin{figure}[!t]
\centering
\includegraphics[width=2.8in]{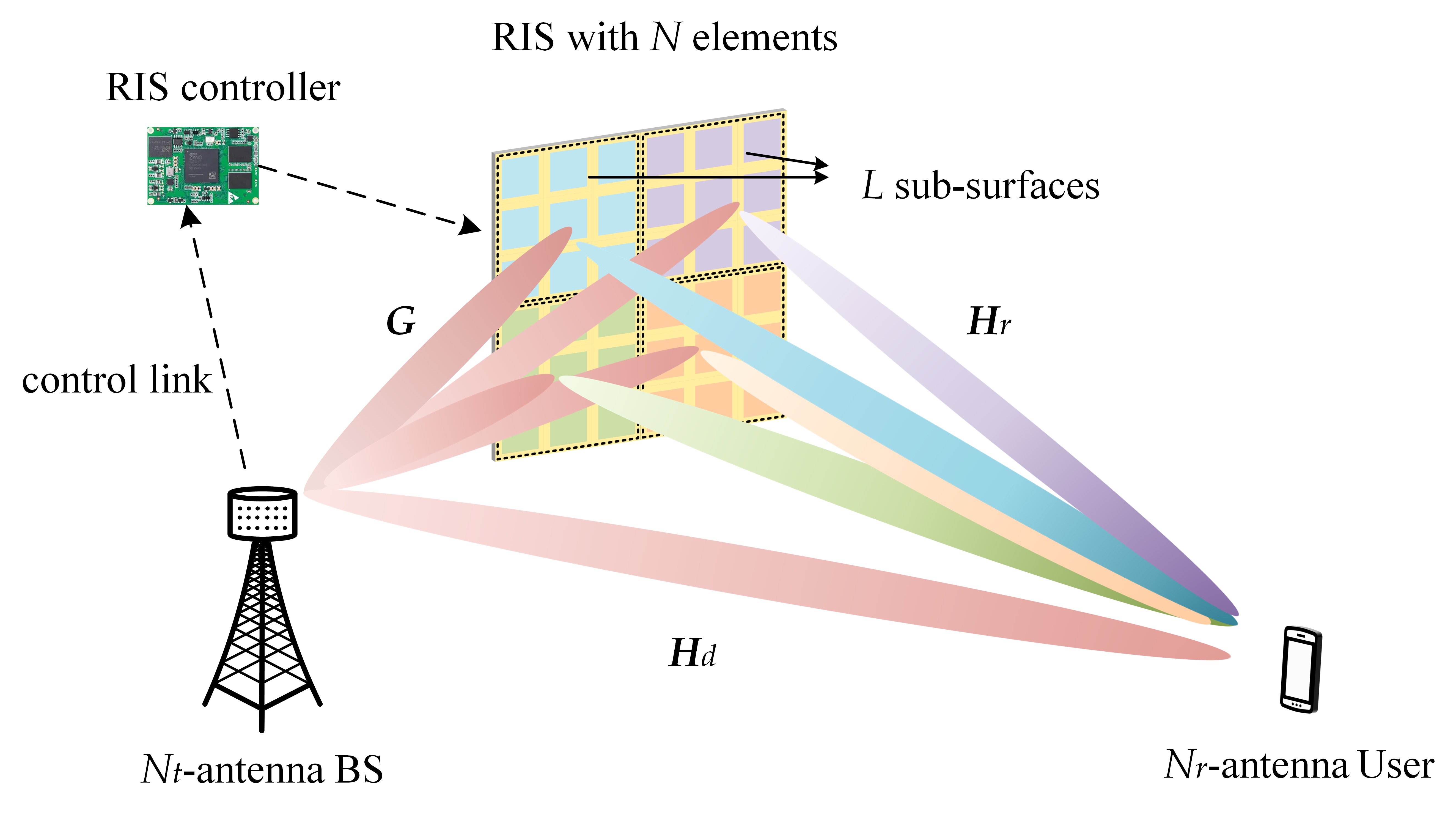}
\caption{An RIS-aided downlink MIMO communication system.}
\label{p1}
\vspace{-0.5cm}
\end{figure}

The channels from BS to the user, from RIS to the user, and from BS to RIS are denoted by $\bm{H}_d\in \mathbb{C}^{N_r\times N_t}$, $\bm{H}_r\in \mathbb{C}^{N_r\times N}$, and $\bm{G}\in \mathbb{C}^{N\times N_t}$, respectively. Since the BS and RIS are close to each other, together with their fixed locations, it is reasonable to assume that the channel between BS and RIS, $\bm{G}$, is a static line-of-sight (LoS) link. On the other hand, the channel between BS and the user and the channel between RIS and the user are modelled as Rayleigh channels for rich scattering in the propagation environment. 
Notably, we consider a general spatially-correlated Rayleigh fading channel. Concretely, we have 
\begin{align}
\bm{H}_d= \sqrt{\beta_d}\bm{R}_u^{1/2} \bm{H}_{\omega}^d \bm{R}_b^{1/2},\enspace \bm{H}_r= \sqrt{\beta_r}\bm{R}_u^{1/2} \bm{H}_{\omega}^r \bm{R}_r^{1/2},
\end{align}
where $\beta_d$ and $\beta_r$ are respectively the large-scale path losses of the direct channel and the channel between RIS and the user, $\bm{R}_b\in \mathbb{C}^{N_t\times N_t}$, $\bm{R}_r\in \mathbb{C}^{N\times N}$, and $\bm{R}_u\in \mathbb{C}^{N_r\times N_r}$ are the spatially correlation matrices at the BS, RIS, and user, respectively.  For mathematical tractability, the correlation matrices at the BS and the user are modeled by the commonly used Kronecker correlation model \cite{channelmodel}. Besides, the components of $\bm{H}_{\omega}^d\in \mathbb{C}^{N_r\times N_t}$ and $\bm{H}_{\omega}^r\in \mathbb{C}^{N_r\times N}$ in (1) are independent and identically distributed (i.i.d.) complex Gaussian matrices with zero mean and unity variance.

We denote the reflection-coefficients matrix at the RIS  by $\bm{\Theta}=\mathrm{diag}\left \{ \varsigma_1 e^{j\theta_1},\cdots,\varsigma_N e^{j\theta_N} \right \}$, where $ \varsigma_n$ and $\theta_n$ represent the amplitude reflection coefficient and the phase shift induced by the $n$th reflecting element. For a typical low-cost passive RIS,  we can normalize the amplitude reflection coefficient and set  $\varsigma_n=1$, $\forall n$. Let $\bm{W}\in \mathbb{C}^{N_t\times N_s}$ represent the precoding matrix with unity power at the BS, i.e., $\mathrm{Tr}\left(\bm{W}\bm{W}^H \right)=1$, and $N_s$ is the number of data streams. Then, we formulate the received signal at the user without extra information transfer via RIS as
\begin{align}\label{eq3}
\bm{y}=\sqrt{P} \left(\bm{H}_d+\bm{H}_r \bm{\Theta} \bm{G}\right )\bm{W}\bm{s} +\bm{z},
\end{align}
where $P$ is the transmit power, $\bm{s}=[s_1,\cdots,s_{N_s}]^T$ is the symbol vector transmitted by the BS, and $\bm{z}$ represents the additive Gaussian noise vector with zero mean and variance $\sigma^2$, i.e., $\bm{z}\sim \mathcal{CN}(\bm{0},\sigma^2 \bm{I}_{N_r})$. The symbol $s_i$ is selected from $M$-ary phase shift keying/quadrature amplitude modulation (PSK/QAM) constellations, satisfying $\mathbb{E}\{\vert s_i\vert^2\}=1$, $\forall i$.\par

\vspace{-0.7cm}
\subsection{Principle of Superimposed RIS-Phase Modulation (SRPM)}

In order to effectively transfer extra messages and achieve higher communication rate, we introduce the principle of the superimposed RIS-phase modulation, referred to as SRPM. Instead of using the RIS only for PBF, we realize  extra information transfer via the selection of superimposed RIS phase shifts. Specifically, in the SRPM scheme, the phase configured at the RIS is made up of two parts, i.e., the base phases and the superimposed phase offsets. The base phase is predetermined and remains unchanged during each symbol period. On the other hand, the superimposed phase offsets is dynamically added to the base phases. The superimposed phase offsets are selected from a specific set, according to the extra information to be transferred from the RIS. Concretely, the phase on the $n$th reflecting element of RIS is configured as
\begin{align}\label{eq4}
\theta_n=\theta_n^{b}+k_n \Delta\theta, 
\end{align}
where $\theta_n^b$ is the predetermined base phase of the $n$th RIS reflecting element, $k_n$ is selected from the set $\mathcal{K}=\{-K,\cdots,0,1,\cdots,K\}$, and $\Delta \theta$ is the unit step  of the superimposed  phase offsets. Defining $K$ as the modulation order for conveying the extra information at RIS, the maximum phase offset is $K\Delta \theta$. Using the superimposed phase in (\ref{eq4}), the received signal in (\ref{eq3}) becomes
\begin{align}\label{eq5}
\bm{y}&=\sqrt{P} \left( \bm{H}_d +\bm{H}_r \mathrm{diag}\{ e^{j\theta_1},\cdots,e^{j\theta_N}\} \bm{G} \right ) \bm{W} \bm{s} +\bm{z}=\sqrt{P} \left( \bm{H}_d +\bm{H}_r \bm{\Theta}^b \bm{\Xi} \bm{G} \right)\bm{W} \bm{s} +\bm{z},
\end{align}
where $\bm{\Theta}^b\triangleq \mathrm{diag}\{ e^{j\theta_1^b},\cdots,e^{j\theta_N^b}\}$ is the base phase shift matrix and $\bm{\Xi}\triangleq {\rm{diag}}\{e^{j k_1 \Delta \theta},\cdots, e^{j k_n \Delta \theta}\}$ is the superimposed phase offset matrix conveying the additional information.

Considering the surface grouping for reducing training overhead and simplifying the RM design, we group the surface elements into $L$ sub-surfaces and each sub-surface contains $N/L$ elements, where $N/L$ is assumed as an integer for simplicity\cite{group}. Note that for $L=N$, this reduces to the conventional setup without RIS grouping. Then, we consider that the reflecting elements in the same sub-surface are configured with the same phase offset. Let $\mathcal{A}_l$ $l\in [1:L]$ represent the set of antenna elements in the $l$th sub-surface and the common phase offset configured in this set is denoted by $k_l\Delta \theta$. Hence, the extra information matrix, $\bm{\Xi}$, becomes $\mathrm{diag}\{ \bm{v} \} \otimes \bm{I}_{\frac{N}{L}}$, where $ \bm{v}\triangleq \left[e^{j k_1 \Delta \theta}, \cdots,e^{j k_L \Delta \theta}\right]^T$ is the information vector transmitted by the RIS. 
We stress that the proposed SRPM scheme is a general framework for RM with RIS. By adjusting the SRPM parameters $K$, $\Delta \theta$, and $\theta_n^b$, we are able to include existing schemes as special cases of the proposed SRPM. For example, in the previous work \cite{my}, basic phases $\theta_n^b$, $\forall n$ are chosen as the conventional optimized phases for PBF. In addition, the proposed SRPM scheme is also effective for more practical discrete phase shifts. Specifically, when discrete phase shifts are quantized by $b$ bits, $\Delta \theta$ can be selected as an integer multiple of $\frac{\pi}{2^{b-1}}$ and $K$ satisfies $K\leq 2^{b-1}$. When we set $\Delta \theta=\frac{\pi}{2^{b-1}}$ and $K= 2^{b-1}$, the proposed SRPM  reduces to the case of RM for theoretical analysis in \cite{theo1}.

At the receiver, we consider to detect $\bm{s}$ and $\bm{v}$ following the principle of the maximum likelihood (ML), which is expressed as
\begin{align}\label{mldetec}
(\hat{\bm{s}},\hat{\bm{v}})=\arg \min_{\bm{s},\bm{v}} \left \Vert \bm{y}-\sqrt{P}  \left( \bm{H}_d +\bm{H}_r \bm{\Theta}^b \bm{\Xi} \bm{G} \right)\bm{W} \bm{s}\right \Vert^2.
\end{align}
Considering the computational complexity, it is necessary to traverse all possible realizations of $\bm{s}$ and $\bm{v}$, leading to the complexity of the ML detector at the order of $\mathcal{O}(M^{N_s} (2K+1)^L)$ \cite{qrm}.

It is worth noting that the proposed SRPM imposes no additional communication overhead between the BS and the RIS controller. The only additional requirement of SRPM is marginal signal processing needed by the RIS controller which is at most limited to searching a small-size look-up table before the common phase-shift control. Hence, the proposed SRPM fully exploits the passive trait of RIS and successfully avoids excessive energy consumption brought by active devices. As for the active RIS case, as long as the reflection coefficient at each RIS element is uniformly amplified, the proposed modulation design is directly applicable.
Furthermore, the SRPM can also be applied to the multiuser systems by using a time division multiple access (TDMA) method \cite{qrm}, where multiple users are activated within different time resource blocks.


\vspace{-0.3cm}
\section{Performance Analysis for the Proposed SRPM scheme}
In this section, we analyze the ABER and ergodic capacity performance of the proposed SRPM scheme with the ML detection, which paves the way for  useful insights. 
\vspace{-0.5cm}
\subsection{ABER Analysis}
To begin with, the conditional pairwise error probability (CPEP) of detecting $(\bm{s},\bm{v})$ from (\ref{mldetec}) as any given $(\hat{\bm{s}},\hat{\bm{v}})$ is expressed by
\vspace{-0.2cm}
\begin{align}
&\Pr \left((\bm{s},\bm{v}) \to (\hat{\bm{s}},\hat{\bm{v}}) \vert \bm {H}_d,\bm{H}_r\right )\nonumber \\
&={\rm{Pr}}\left( \left\Vert \bm{y}-\sqrt{P}\left(\bm{H}_d+\bm{H}_r \bm{\Theta}^b \bm{\Xi} \bm{G}\right )\bm{W}\bm{s} \right \Vert^2 >\left\Vert \bm{y}-\sqrt{P}\left(\bm{H}_d+\bm{H}_r \bm{\Theta}^b \hat{\bm{\Xi}} \bm{G}\right )\bm{W}\hat{\bm{s}} \right \Vert^2  \right)\nonumber \\
&\overset{(a)}{=}{\rm{Pr}}\left( \left\Vert \bm{z} \right \Vert^2 >\left\Vert \bm{z}+\sqrt{P}\left (\left(\bm{H}_d+\bm{H}_r \bm{\Theta}^b \bm{\Xi} \bm{G}\right )\bm{W}\bm{s}-\left(\bm{H}_d+\bm{H}_r \bm{\Theta}^b \hat{\bm{\Xi}} \bm{G}\right )\bm{W}\hat{\bm{s}}
 \right)\right \Vert^2 \right)\nonumber \\
&=\Pr\left( -P \Vert \bm{\delta} \Vert^2 -2\sqrt{P }{\rm{Re}}\left\{ \bm{\delta}^H \bm{z} \right\}>0\right),
\end{align}
\vspace{-0.2cm}
where $(a)$ comes from (\ref{eq5}), $\bm{z}$ is the additive Gaussian noise vector, $\hat{\bm{\Xi}}\triangleq \mathrm{diag}\left \{\hat{\bm{v}}\right\} \otimes \bm{I}_{N/L}$, and 
\begin{align}\label{delta}
\bm{\delta}\triangleq \left(\bm{H}_d+\bm{H}_r \bm{\Theta}^b \bm{\Xi} \bm{G}\right )\bm{W}\bm{s}-\left(\bm{H}_d+\bm{H}_r \bm{\Theta}^b \hat{\bm{\Xi}} \bm{G}\right )\bm{W}\hat{\bm{s}}.
\end{align}
Considering that the channels are given, we find that $  -P \Vert \bm{\delta} \Vert^2 -2\sqrt{P }{\rm{Re}}\left\{ \bm{\delta}^H \bm{z} \right\}$ is a real Gaussian random variable with mean $ -P \Vert \bm{\delta} \Vert^2 $ and variance $2 P \sigma^2 \Vert \bm{\delta} \Vert^2$. Hence, according to \cite[Eq.~(A.37)]{fun}, we obtain
\begin{align}\label{equa6}
\Pr\left ((\bm{s},\bm{v}) \to (\hat{\bm{s}},\hat{\bm{v}}) \vert \bm {H}_d,\bm{H}_r\right)=\mathcal{Q} \left(\sqrt{\frac{P\Vert \bm{\delta}\Vert^2 }{2\sigma^2}} \right)=\mathcal{Q} \left(\sqrt{\frac{P   \lambda}{2\sigma^2}} \right).
\end{align}
where $\mathcal{Q}(\cdot)$ is the Gaussian Q-function \cite[Eq. (2)]{approx} and $\lambda \triangleq \Vert \bm{\delta} \Vert^2$. Then, in order to derive the average pairwise error probability (APEP), we need to first get the distribution of $\lambda$. For any fixed channel $\bm{G}$ as well as the corresponding $\bm{W}$ and $\bm{\Theta}^b$, we derive the moment-generating function (MGF) of $\lambda$ in \emph{Lemma 1}.
\vspace{-0.3cm}
\begin{lemma}  The MGF of $\lambda$ is 
\begin{align}\label{new8}
\mathcal{M}_\lambda(t)=\prod_{i=1}^{N_r} (1-\lambda_i t)^{-1},
\end{align}
where $\{\lambda_i\}_{i=1}^{N_r}$ are defined in Appendix A. Particularly, for ideal uncorrelated channels, variable $\lambda$ is a Chi-squared random variable with the degrees of freedom $2N_r$ and its MGF reduces to
\vspace{-0.3cm}
\begin{align}
\mathcal{M}_\lambda(t)= (1-\sigma_d^2 t)^{-N_r},
\end{align}
\vspace{-0.3cm}
where $\sigma_d^2$ is a constant defined in Appendix A.
\end{lemma}
\vspace{-0.3cm}
\begin{proof}
Please refer to Appendix A. $\hfill\square$
\end{proof}

\vspace{-0.3cm}
Then by applying (\ref{new8}) in (\ref{equa6}), the APEP of SRPM is further rewritten as
\vspace{-0.3cm}
\begin{align} \label{apep}
&\Pr\left((\bm{s},\bm{v}) \to (\hat{\bm{s}},\hat{\bm{v}})\right)=\mathbb{E}\left\{ \Pr\left((\bm{s},\bm{v}) \to (\hat{\bm{s}},\hat{\bm{v}}) \vert \bm {H}_d,\bm{H}_r\right ) \right \} =\int_{0}^{+\infty} \mathcal{Q} \left(\sqrt{\frac{P \lambda}{2\sigma^2}} \right) f_{\lambda} (\lambda)\, \mathrm{d}\lambda\nonumber \\
&\quad \overset{(a)}{\approx} \int_{0}^{+\infty} \left(\frac{1}{12}e^{-\frac{P \lambda}{4\sigma^2}}+\frac{1}{4}e^{-\frac{P \lambda}{3\sigma^2}} \right) f_{\lambda} (\lambda) \enspace \mathrm{d}\lambda =\frac{1}{12}\mathcal{M}_\lambda \left(-\frac{P  }{4\sigma^2} \right)+\frac{1}{4}\mathcal{M}_\lambda \left(-\frac{P}{3\sigma^2} \right)\nonumber\\
&\quad =\frac{1}{12}\prod_{i=1}^{N_r}\left(1+\frac{P\lambda_i}{4\sigma^2} \right)^{-1}+\frac{1}{4}\prod_{i=1}^{N_r}\left(1+\frac{P\lambda_i}{3\sigma^2} \right)^{-1},
\end{align}
where $f_{\lambda}(\cdot)$ is the probability distribution function (PDF) of $\lambda$, and $(a)$ exploits the fact that the Q-function is well approximated as $\mathcal{Q}(x)\approx \frac{1}{12}e^{-\frac{x^2}{2}}+\frac{1}{4}e^{-\frac{2x^2}{3}}$ \cite{approx}. It is worth noting that this approximation is tight for large $x$ which therefore ensures its effectiveness and accuracy in evaluating the diversity order.

Now, by accumulating the APEP for all possible cases, we obtain a union bound of ABER~as
\vspace{-0.4cm} 
\begin{align} \label{eq17}
P_b\leq& \frac{1}{M^{N_s} (2K+1)^L}\sum_{\bm{s}} \sum_{\bm{v}} \sum_{\hat{\bm{s}}} \sum_{\hat{\bm{v}}}{\rm{Pr}}((\bm{s},\bm{v}) \to (\hat{\bm{s}},\hat{\bm{v}})) \frac{e((\bm{s},\bm{v}) \to (\hat{\bm{s}},\hat{\bm{v}}))}{N_s\log_2 M+L\lfloor\log_2(2K+1)\rfloor},
\end{align}
where $e((\bm{s},\bm{v}) \to (\hat{\bm{s}},\hat{\bm{v}}))$ stands for the number of error bits when detecting $(\bm{s},\bm{v})$ as $(\hat{\bm{s}},\hat{\bm{v}})$. The number of error bits is upto $\log_2 M$ for each misestimated $\hat{s}_i$, and  the number of error bits is upto $\lfloor\log_2(2K+1)\rfloor$ for each misestimated $\hat{v}_l$. Thus, $e(\bm{s},\bm{v} \to \hat{\bm{s}},\hat{\bm{v}})$ is calculated by accumulating all these possible error bits. By analyzing (\ref{eq17}), the diversity order of the proposed SRPM is revealed in \emph{Theorem~1}.
\vspace{-0.3cm}
\begin{theorem}
Given the precoding at BS and base phase shifts at RIS, the diversity order of the proposed SRPM is equal to $N_r$, which is optimal for all legitimate system parameter settings.
\end{theorem}
\begin{proof}
Considering the asymptotic case with large transmit power, i.e., $P/\sigma^2 \to \infty$, we derive the diversity order which by checking that
\begin{align}
\lim\limits_{P\to \infty} -\frac{\log_2 \Pr\left((\bm{s},\bm{v})\to (\hat{\bm{s}} ,\hat{\bm{v}})\right)}{\log_2 P}=N_r.
\end{align}
which holds if the channel is not deficient. Therefore, we conclude that the diversity order of the proposed SRPM scheme is equal to $N_r$ regardless of the other parameter settings of, e.g., $L$, $M$, $N_s$, $K$, and $\Delta \theta$. The proof completes. \hfill $\square$
\end{proof}

\vspace{-0.4cm}
Following are some immediate remarks that can be useful for RIS deployment in practice.
\vspace{-0.4cm}

\begin{remark}
Compared with \cite{my}, it is observed that a higher diversity order can be achieved by arbitrarily given precoding vector and base phase shifts, rather than the optimal ones for PBF. This is because the optimal phase shift rotates the channel coefficients into real numbers, which maximizes the SNR at the expense of the loss of DoF. To be specific, we take a single-input single-output (SISO) system without direct link between BS and RIS as an example. The received signal in (\ref{eq5}) reduces to
$y=\beta_r \sum_{l=1}^L e^{j k_l\Delta \theta} s\sum_{n\in \mathcal{A}_l}h_{r,n} e^{j \theta_n^b}g_n+z$,
where $h_{r,n}\sim\mathcal{CN}(0,1)$ and $g_n$ is a constant. To maximize the SNR, the optimal base phase shifts for PBF is selected as $\theta_n^b =-\angle( h_{r,n} g_n)$ for all $n$ \cite{wu}. Meanwhile, the random variable $\sum_{n\in \mathcal{A}_l}h_{r,n} e^{j \theta_n^b}g_n$ becomes a real value, which is approximated by a Gaussian variable \cite{my}. Therefore, the variable $\lambda$ in (\ref{equa6}) reduces to a Chi-squared random variable with the degrees of freedom {1}. On the contrary, with a fixed $\theta_n^b$, the random variable $\sum_{n\in \mathcal{A}_l}h_{r,n} e^{j \theta_n^b}g_n$ is a complex Gaussian variable and variable $\lambda$ in (\ref{equa6}) is  with the degrees of freedom {2}, which explains the increase in diversity order. Similar to the trade-off between diversity gain and multiplexing gain, selecting the optimal phase shifts for PBF achieves better ABER at low SNR levels, while an arbitrary design of fixed phase shifts exhibits better ABER at high SNR regimes. 
\end{remark}
\vspace{-0.5cm}
\begin{remark}
It is worth noting that, compared with existing works, such as PBF\cite{pbf2,pbf3}, SM \cite{rsm,tsm,gsm}, and RM \cite{pbit}\cite{rpm}\cite{qrm}, the proposed SRPM scheme is more appealing to practical cases where only partial CSI is available at the transmitter and the phase shifts are discrete at RIS. Note that the diversity order is not affected by the precoding matrix and the base phase shifts. In other words, the inaccurate CSI at the transmitter and the discrete phase shifts at RIS do not harm the diversity order.
\end{remark}
\vspace{-0.5cm}
\begin{remark}
It is found that the ABER relates to the SRPM parameters, including $L$, $M$, $N_s$, $K$, and $\Delta \theta$, which can be optimized for performance gains. For the optimization of $(L,M,N_s,K,\Delta\theta)$, especially for the discrete phase shifts, an exhaustive search is able to find the optimal solution. 
\end{remark}


\vspace{-0.5cm}

\subsection{Ergodic Capacity Analysis}
Considering a typical case where discrete constellations are exploited to transmit messages with the proposed SRPM, we evaluate the discrete-input continuous-output memoryless channel (DCMC) capacity, which is defined as \cite{capacity1,capacity2}
\vspace{-0.3cm}
\begin{align}
C_{\mathrm{SRPM}}=\mathbb{E}\left \{ \max_{p_{\bm{s}}(\bm{s}),\,p_{\bm{v}}(\bm{v})} \mathcal{I}(\bm{s},\bm{v};\bm{y}) \right \},
\end{align}
where $p_{\bm{s}}(\cdot)$ and $p_{\bm{v}}(\cdot)$ are the probability distributions of independent random variables $\bm{s}$ and $\bm{v}$, respectively, and $\mathcal{I}(\bm{s},\bm{v};\bm{y})$ is the mutual information between the symbols, $(\bm{s},\bm{v})$, and the received signal, $\bm{y}$. The DCMC capacity of SRPM is derived in the following theorem.
\vspace{-0.3cm}
\begin{theorem}
The DCMC capacity of SRPM is characterized as
\vspace{-0.3cm}
\begin{align}\label{capacity}
C_{\mathrm{SRPM}}&\approx 2\log_2 S-\log_2 \left(S+\sum_{\bm{s}} \sum_{\bm{v}}\sum_{\hat{\bm{s}}\neq \bm{s}} \sum_{\hat{\bm{v}}\neq \bm{v}} \mathcal{M}_{\lambda}\left(-\frac{P}{2\sigma^2}\right)\right),
\end{align}
where $S$ is defined in Appendix B.
\end{theorem}
\vspace{-0.3cm}
\begin{proof}
Please refer to Appendix B.\hfill $\square$
\end{proof}
\vspace{-0.3cm}
To obtain more insightful observations from (\ref{capacity}), the asymptotic capacity limit of the proposed SRPM tends to
\vspace{-0.3cm}
\begin{align}
\lim_{P\to \infty} C_{\mathrm{SRPM}}=\log_2 S=N_s\log_2 M +L\log_2(2K+1).
\end{align}
On the other hand, the high SNR capacity limit of the traditional PBF equals to $N_s\log_2 M$ \cite[Eq. (42)]{theo1}. It is obvious that in the high SNR regime, the proposed SRPM increases the achievable rate by additional $L\log_2(2K+1)$ bits per channel use (bpcu) as compared to the traditional PBF strategy with the same modulation order at the BS.

\vspace{-0.4cm}
\section{Precoding Optimization for Single-stream Cases}
\vspace{-0.3cm}
Based on the derived results in Section \Rmnum{3}, we are able to optimize the precoding at BS  to further improve  the ABER and ergodic performance. In this section, we consider the single-stream case, i.e., $N_s=1$, and the ideal uncorrelated channels, which facilitates the optimal design for the precoding vector.

Firstly, we reformulate the ABER as a function of the precoding, $\bm{w}$, in the following lemma.
\vspace{-0.7cm}
\begin{lemma}
For the single-stream case with uncorrelated channels, the ABER of the proposed SRPM scheme is rewritten as
\begin{align}\label{eq18}
P_b \leq \frac{1}{rS}\sum_{s} \sum_{\bm{v}} \sum_{\hat{s}} \sum_{\hat{\bm{s}}} e((s,\bm{v})\to (\hat{s},\hat{\bm{v}}))\sum_{i=1}^Q a_i \left( 1+\frac{P b_i}{2\sigma^2}\left(\beta_d \vert s-\hat{s}\vert^2 +\bm{w}^H \bm{A}_{(s,\bm{v})\to (\hat{s},\hat{\bm{v}})}\bm{w} \right)\right)^{-N_r}.
\end{align}
\vspace{-0.3cm}
where $r\triangleq N_s\log_2(M)+L\lfloor\log_2(2K+1)\rfloor$, $Q$, $a_i$, $b_i$, and $\bm{A}_{(s,\bm{v}) \to (\hat{s},\hat{\bm{v}})}$ are defined in Appendix~C.
\end{lemma}
\vspace{-0.3cm}
\begin{proof}
Please refer to Appendix C. $\hfill\square$
\end{proof}
\vspace{-0.3cm}
\begin{remark}
By increasing $Q$, the upper bounds tend to be tighter. Hence, we can choose $Q$ according to the SNR level to accurately characterize the ABER.
\end{remark}
\vspace{-0.3cm}
\begin{remark}
Note that for ABER, the working area of practical interest is always at relatively high SNRs corresponding to satisfactorily low ABER. In practice, an absolutely low SNR can be regarded as relatively high for a low-order modulation that better fits the poor wireless channel. Hence, it is reasonable to consider the ABER at high SNRs, where $Q=1$ is sufficient to accurately capture the ABER.
\end{remark}
\vspace{-0.3cm}
\begin{corollary}
According to (\ref{capacity}), the ergodic capacity of the proposed SRPM can also be expressed as a function of $\bm{w}$, i.e.,
\vspace{-0.3cm}
\begin{align}\label{eq19}
\mathcal{C}_{\mathrm{SRPM}} \approx 2\log_2 S-\log_2\left( S+\sum_{s} \sum_{\bm{v}} \sum_{\hat{s}} \sum_{\hat{\bm{s}}} \left( 1+\frac{P}{2\sigma^2}\left(\beta_d \vert s-\hat{s}\vert^2 +\bm{w}^H \bm{A}_{(s,\bm{v})\to (\hat{s},\hat{\bm{v}})}\bm{w} \right)\right)^{-N_r}\right).
\end{align}
\end{corollary}
\vspace{-0.5cm}
From (\ref{eq18}) and (\ref{eq19}), we are able to establish a general optimization framework for ABER minimization and ergodic capacity maximization. The universal precoding optimization problem follows
\vspace{-0.4cm}
\begin{align} \label{problem}
\mathop{{\rm{minimize}}}_{\bm{w}} &\quad \sum_{s} \sum_{\bm{v}} \sum_{\hat{s}} \sum_{\hat{\bm{v}}}\left( 1+\frac{P}{2\sigma^2}\left(\beta_d\vert s-\hat{s} \vert^2+\bm{w}^H\bm{A}_{(s,\bm{v}) \to (\hat{s},\hat{\bm{v}})}\bm{w} \right)\right)^{-N_r}\nonumber \\
\mathrm{subject}\enspace \mathrm{to}&\quad \Vert \bm{w} \Vert^2=1. 
\end{align}
Through some linear transformations, the objective function in (\ref{problem}) is directly extended to (\ref{eq18}) and (\ref{eq19}). Note that the problem in (\ref{problem}) is a nonconvex QCQP problem, which is in general NP-hard and intractable. In order to remove the quadratic term, we introduce that $\bm{W}=\bm{w} \bm{w}^H$ and equivalently rewrite the problem in (\ref{problem}) as
\vspace{-0.3cm}
\begin{align}\label{problem2}
\mathop{{\rm{minimize}}}_{\bm{W}} &\quad  f(\bm{W})\nonumber \\
\mathrm{subject}\enspace \mathrm{to}&\quad \mathrm{Tr}(\bm{W})=1, \enspace \bm{W}\succeq \bm{0}, \enspace \mathrm{rank}(\bm{W})=1,
\end{align}
where $ f(\bm{W})\triangleq \sum_{s} \sum_{\bm{v}} \sum_{\hat{s}} \sum_{\hat{\bm{v}}}\left( 1+\frac{P}{2\sigma^2}\left(\beta_d\vert s-\hat{s} \vert^2+\mathrm{Tr}\left(\bm{A}_{(s,\bm{v} )\to (\hat{s},\hat{\bm{v}})}\bm{W} \right)\right)\right)^{-N_r}$. Note that  the problem in (\ref{problem2}) is still nonconvex due to the rank-one constraint. To circumvent this difficulty, the SDR technique has been widely used to relax the nonconvex rank-one constraint \cite{sdr1}\cite{sdr2}. By exploiting the SDR, we transform the problem in (\ref{problem2}) to the following relaxed problem as
\vspace{-0.3cm}
\begin{align}\label{problem3}
\mathop{{\rm{minimize}}}_{\bm{W}} &\quad f(\bm{W})\nonumber \\
\mathrm{subject}\enspace \mathrm{to}&\quad \mathrm{Tr}(\bm{W})=1, \enspace \bm{W}\succeq \bm{0},
\end{align}
which is convex and thus can be efficiently solved by existing convex program solvers, e.g., CVX tools \cite{cvx}. 
According to \cite{luoz}, the worst-case complexity of solving the above problem via an interior-point algorithm is at the order of $\mathcal{O}\left( \max \left\{S^2, N_t \right \}^4 {N_t}^{0.5}  \log(1/\epsilon) \right)$, where $\epsilon$ is required solution accuracy. Notably, the complexity grows polynomially with respect to the modulation order, $S$, and the number of antennas at the BS, $N_t$, but not the number of RIS elements.

Although the rank-one constraint is relaxed, we prove in the following theorem that the relaxation is fortunately tight for this problem.
\vspace{-0.3cm}
\begin{theorem}
For a feasible problem in (\ref{problem3}), an optimal $\bm{W}$ with rank one is always available.
\end{theorem}
\vspace{-0.3cm}
\begin{proof}
Please refer to Appendix D. \hfill $\square$
\end{proof}
\vspace{-0.3cm}

Based on \emph{Theorem 3}, we verify that the optimal $\bm{w}^{\mathrm{opt}}$ can be obtained via the eigenvalue decomposition (EVD) of $\bm{W}^{\mathrm{opt}}$. Note that the precoding optimization is based on the CSI of the channel between BS and RIS rather than the full CSI of all channels. The channel between BS and RIS is assumed as long-term static, which can be estimated with high accuracy at low cost. 
Specifically, the user can first transmit pilots to the BS with different phase shift matrices at the RIS and then channel measurements are received by the BS to perform channel estimation. Channel estimation algorithms, like the ones proposed in \cite{cest2,cest3}, are available for estimating the channel between the BS and RIS.
Hence, the proposed precoding optimization is friendly and applicable especially for resource-limited systems in practice.

We then consider a special case where a closed-form optimal solution of the precoding vector is admitted. Assume that there exists a major line-of-sight (LoS) path between BS and RIS, that is, $\bm{G}$ can be well represented by a matrix with rank one, which is further expressed as
$\bm{G}=\beta \textbf{a}_N(\phi_r) \textbf{a}_{N_t}^H(\phi_t)$,
where $\beta$ is the complex gain of the path, $\textbf{a}_N(\phi_r)$ and $\textbf{a}_{N_t}(\phi_t)$ respectively represent the receive and transmit steering vectors, and $\phi_r$ and $\phi_t$ denote the angle of arrival (AoA) at RIS and the angle of departure (AoD) at BS, respectively.
The variable $\beta$ is complex-valued and follows the distribution $\mathcal{CN}(0,10^{-0.1\kappa})$ \cite{wangp} where $\kappa$ represents the path loss coefficient. The steering vector is defined as
\vspace{-0.3cm}
\begin{align}\label{eq2}
\textbf{a}_N(\vartheta)&=\left [1,e^{j\frac{2\pi d}{\lambda}\sin\vartheta },\cdots,e^{j\frac{2\pi d}{\lambda}(N-1)\sin \vartheta }\right]^T, 
\end{align}
where $d$ is the antenna separation space distance, and $\lambda$ is the wavelength of the carrier. By defining $\bm{a}\triangleq \textbf{a}_N(\phi_r)$ and $\bm{b}\triangleq \textbf{a}_{N_t}(\phi_r)$, we have the following lemma.
\vspace{-0.4cm}
\begin{lemma}
For a keyhole channel $\bm{G}=\beta \bm{a} \bm{b}^H$, the optimal precoding vector, $\bm{w}^{\mathrm{opt}}$, for the proposed SRPM scheme is given by
$\bm{w}^{\mathrm{opt}}=\frac{\bm{b}}{\Vert \bm{b}\Vert}$.
\end{lemma}
\vspace{-0.4cm}
\begin{proof}
Please refer to Appendix E. $\hfill\square$
\end{proof}
\vspace{-0.4cm}
\begin{remark}
It is found that the optimal beam direction is consistent with the transmit steering vector at BS, rather than the equivalent channel between the BS and the user in existing PBF schemes like \cite{wu}. This is because aligning the beam to a fixed user direction for a long period of time can cause performance degradation due to the dynamic user-related channels.
\end{remark}

\vspace{-0.5cm}
\section{Low-complexity Sphere-Decoding Based Layered Detector Design}
\vspace{-0.2cm}
Due to the extremely high complexity of ML detection, it is necessary to develop a low-complexity detector to promote the development of the proposed SRPM scheme in practice. Sphere decoding is a promising MIMO detection strategy with near-ML performance and low computational complexity \cite{sd1,sd2,sd3}. In this section, we design an SD-based layered detector to achieve efficient and low computational complexity detection for SRPM.

To pave the way for the development of a low-complexity detector, we first rewrite the received signal in (\ref{eq5}) to a conciser form. By defining $\bm{A}\triangleq \bm{H}_d \bm{W}$ and $\bm{B}\triangleq \bm{G} \bm{W}$, we have 
\vspace{-0.3cm}
\begin{align}
\bm{y}=\sqrt{P}\left( \bm{A} \bm{s}+ \bm{H}_r \bm{\Theta}^b \bm{\Xi}\bm{B}\bm{s}\right ) +\bm{z},
\end{align}
where $\bm{A}$, $\bm{H}_r\bm{\Theta}^b$, and $\bm{B}$ are equivalent channels, and $\bm{\Xi}$ and $\bm{s}$ are the symbols to be detected. Specifically, the received signal at the $i$th antenna is
\vspace{-0.3cm}
\begin{align}
y_i&=\sqrt{P} \sum_{m=1}^{N_s} a_{i,m} s_m +  \sqrt{P}\sum_{l=1}^L e^{j k_l \Delta \theta} \sum_{n\in \mathcal{A}_l} \sum_{m=1}^{N_s} h_{r,i,n}e^{j\theta_n^b} b_{n,m} s_m+z_i\nonumber \\
&=\sqrt{P} \sum_{m=1}^{N_s} a_{i,m} s_m +  \sqrt{P} \sum_{m=1}^{N_s} \sum_{l=1}^L c_{i,l,m} e^{j k_l \Delta \theta} s_m+z_i
\end{align}
where $a_{i,m}$ is the $(i,m)$th element of $\bm{A}$, $h_{r,i,n}$ is the $(i,n)$th element of $\bm{H}_r$, $b_{n,m}$ is the $(n,m)$th element of $\bm{B}$, and $c_{i,l,m}\triangleq \sum_{n\in \mathcal{A}_l} h_{r,i,n}e^{j\theta_n^b} b_{n,m}$. Defining $\bm{h}_{i,l}\triangleq [c_{i,l,1},\cdots,c_{i,l,N_s}]^H$, $l=1,\cdots,L$, $\bm{h}_{i,L+1}\triangleq [a_{i,1},\cdots,a_{i,N_s}]^H$,  $\bm{h}_{i}=[\bm{h}_{i,1}^H,\cdots,\bm{h}_{i,L+1}^H ]^H$, and $\bar{\bm{v}}=[e^{j k_1\Delta \theta},\cdots,e^{j k_L\Delta \theta},1]^T$, we have
\vspace{-0.3cm}
\begin{align}
y_i=\sqrt{P}\sum_{l=1}^{L+1} e^{j k_l \Delta \theta} \bm{h}_{i,l}^H \bm{s} +z_i=\sqrt{P}\bm{h}_i^H \left( \bar{\bm{v}} \otimes \bm{s} \right)+z_i.
\end{align}
Hence, the received signal vector, $\bm{y}$, is equivalently reformulated as
\vspace{-0.3cm}
\begin{align}\label{eq9}
\bm{y}=\sqrt{P} \left[\bm{h}_1,\cdots,\bm{h}_{N_r}\right]^H\left( \bar{\bm{v}} \otimes \bm{s} \right)+\bm{z}=\sqrt{P}\bm{H}\bm{x}+\bm{z},
\end{align}
where $\bm{H}\triangleq \left[\bm{h}_1,\cdots,\bm{h}_{N_r}\right]^H$ is the equivalent channel and $\bm{x} \triangleq \bar{\bm{v}} \otimes \bm{s}$ is the equivalent symbol. 

Note that some linear decoders, such as zero-forcing (ZF) and minimum mean squared error (MMSE) detectors, can be used to detect the equivalent symbol vector $\bm{x}$. Then, we can reconstruct $\bm{s}$ and $\bm{v}$ based on the detected $\bm{x}$. However, these linear detection methods do not take the unique structure of $\bm{x}$ into consideration, which usually lead to significant performance loss. Hence, we resort to the SD technique to develop an effective detection algorithm.

Following the principle of SD, instead of searching all possible symbols, we reduce the computational complexity through  restricting the search within a hypersphere, i.e. satisfying
\begin{align}\label{eq27}
\left \Vert \bm{y}-\sqrt{P} \bm{H}\bm{x}\right \Vert^2 \leq D^2,
\end{align}
where $D$ is a predetermined search radius. To facilitate the implementation of SD,  assume that $N_r>(L+1)N_s$, i.e., sufficient receive antennas are deployed. Then, we can rewrite the equivalent channel $\bm{H}$ in (\ref{eq9}) by applying the QR factorization as
\begin{align}
\bm{H}=\left [\bm{Q}_1,\bm{Q}_2 \right ]\left[\bm{R},\bm{0}\right]^T,
\end{align}
where both $\bm{Q}_1\in \mathbb{C}^{N_r\times (L+1)N_s}$ and $\bm{Q}_2\in \mathbb{C}^{N_r\times N_r- (L+1)N_s}$ are column-orthogonal matrices, and $\bm{R}\in \mathbb{C}^{(L+1)N_s\times (L+1)N_s}$ is an upper triangular matrix. Next, we rewrite (\ref{eq27}) as
\begin{align}\label{eq30}
\left \Vert \bar{\bm{y}}-\sqrt{P} \bm{R}\bm{x}\right \Vert^2=\sum_{i=1}^{(L+1)N_s}\left \vert \bar{y}_i-\sum_{j=i}^{(L+1)N_s}R_{i,j} x_j  \right \vert\leq \bar{D}^2,
\end{align}
where $\bar{\bm{y}}\triangleq \bm{Q}_1^H  \bm{y}$, $\bar{D}^2\triangleq D^2-\left \Vert \bm{Q}_2^H \bm{y}\right \Vert^2$, $R_{i,j}$ is the $(i,j)$th element of $\bm{R}$, and $x_j$ is the $j$th element of $\bm{x}$.

Note that the equivalent symbol vector $\bm{x}$ of $(L+1) N_s$-dimension only contains $L$ BS-transmitted symbols and $N_s$ phase-modulated symbols,  and the traditional SD procedure cannot be directly applied to the SRPM detection. Therefore, we divide the  SD procedure into two layers to detect the BS-transmitted symbols and phase-modulated symbols respectively. In the first layer, given that $\bm{x}_{[LN_s+1:(L+1)N_s]}=\bm{s}$, we first determine the BS-transmitted symbols. According to the philosophy of SD, the search starts from $i=(L+1)N_s$ and the feasible region of $x_{i}$ is reduced by a necessary condition of (\ref{eq30}), i.e.,
\begin{align}
\left \vert \bar{y}_{(L+1)N_s}-R_{(L+1)N_s,(L+1)N_s}x_{(L+1)N_s}\right \vert \leq \bar{D}^2. 
\end{align}
Once we select a feasible candidate of $x_{i}$, we set $i=i-1$ and proceed with the search. At this time, the sphere constraint of $x_i$ is expressed as
\begin{align}\label{sd1}
\left \vert \bar{y}_{i}-\sum_{j=i}^{(L+1)N_s}R_{i,j} x_j  \right \vert\leq {D}_i^2,
\end{align}
where $D_i^2=D_{i+1}^2-\left \vert \bar{y}_{i+1}-\sum_{j=i+1}^{(L+1)N_s}R_{i+1,j} x_j  \right \vert$ and $D_{(L+1)N_s}^2=\bar{D}^2$. This process is terminated when $i=LN_s$ and $\bm{s}$ is detected.

For the second layer, i.e., $i\leq L N_s$, since $\bm{s}$ is determined, every $N_s$ elements of $\bm{x}$ contain the same phase-modulated symbol. Hence, we start from $i^{\prime}=L$ and the sphere constraints for $i^{\prime}$ is formulated as
\begin{align}\label{sd2}
\sum_{i=N_s i^{\prime}-N_s+1}^{N_s i^{\prime}}\left \vert \bar{y}_{i}-\sum_{j=i}^{(L+1)N_s}R_{i,j} x_j  \right \vert \leq \tilde{D}_{i^{\prime}}^2,
\end{align}
where $\tilde{D}_{i^{\prime}}^2=\tilde{D}_{i^{\prime}+1}^2-\sum_{i=N_s( i^{\prime}-1)+1}^{N_s i^{\prime}}\left \vert \bar{y}_{i+1}-\sum_{j=i+1}^{(L+1)N_s}R_{i+1,j} x_j  \right \vert$ and $\tilde{D}_{L}^2=D_{LN_s}^2$. Repeating the above process down to $i^{\prime}=0$, we can find a pair of symbols $(\bm{s},\bm{v})$ that satisfy (\ref{eq30}). More intuitively, we represent the above search domain as a tree, as depicted in Fig. \ref{p2}.

\begin{figure}[!t]
\centering
\includegraphics[width=3.2in]{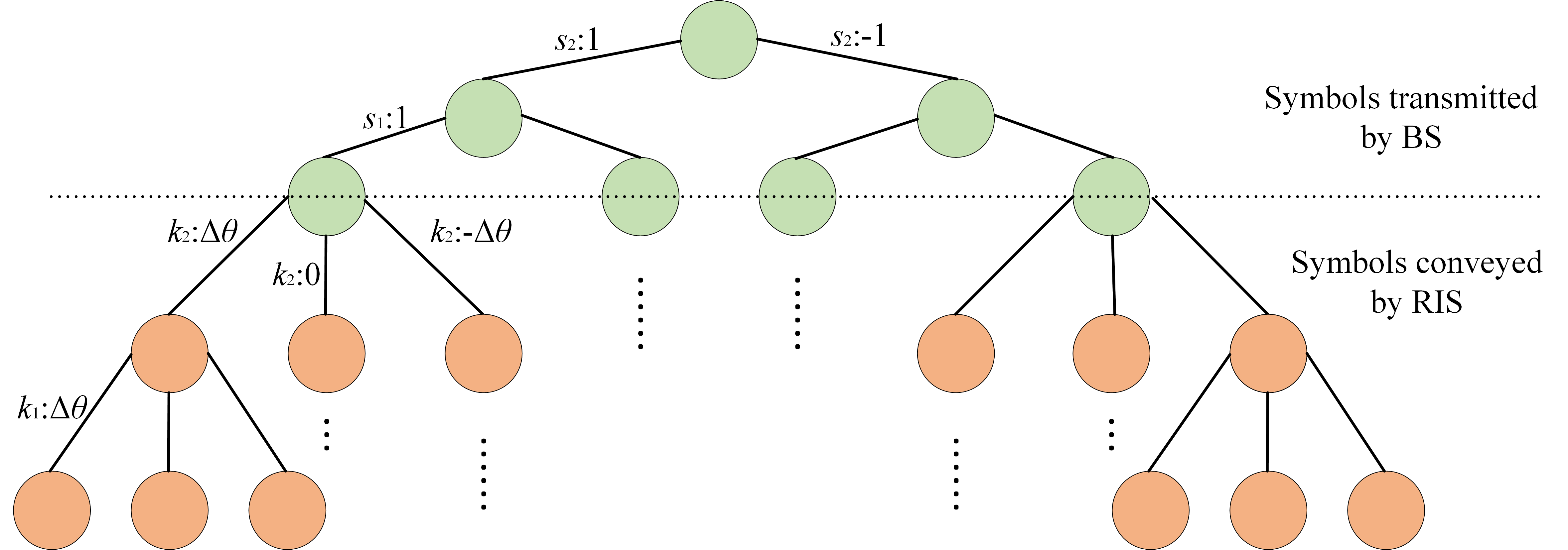}
\caption{An example of the search tree of the proposed SD-based detector for SRPM with $N_s=2$, $M=2$, $L=2$, and $K=1$.}
\label{p2}
\end{figure}

To summarize, we describe the proposed SD-based layered detection algorithm for SRPM in Algorithm~1. According to \cite{sd3}, the expected complexity of the SD-based layered detector is given by $\mathcal{O}\left (M^{\xi N_s}(2K+1)^{\xi L}\right )$, where $\xi \in (0,1]$ is a small factor depending on the SNR levels and satisfies $\xi\ll 1$ at high SNR regimes. Moreover specifically, when SNR increases by 6 dB, the value of $\xi$ becomes half of the original\cite{sd3}.

\begin{spacing}{1.1}
\begin{algorithm}\small
    \caption{SD-based layered detection algorithm for the proposed SRPM scheme}\label{algo2}
    \begin{algorithmic}[1]
		\STATE \textbf{Initialize} $k=L+N_s$, $D_{\mathrm{min}}^2=+\infty$, calculate $\bm{R}$, $\bar{\bm{y}}$, and $\bar{D}^2$ in (\ref{eq30}).
		\WHILE{$k\leq L+N_s$}
		\IF{$k>L$}
			\IF {all possible symbols of $s_{L(N_s-1)+k}$ are traversed} 
			\STATE Set $k=k+1$.
			\ELSE
			\STATE Randomly choose an unselected BS-transmitted symbol as $s_{L(N_s-1)+k}$ and set $i=L(N_s-1)+k$.
			\IF {(\ref{sd1}) is valid}
				\STATE Set $k=k-1$ and update $D_{L(N_s-1)+k-1}^2$.
			\ENDIF
			\ENDIF

		\ELSE
			\IF {all possible symbols of $v_k$ are traversed}
				\STATE Set $k=k+1$.
			\ELSE
			\STATE Randomly choose an unselected RIS modulation symbol as $v_k$ and set $i^{\prime}=k$.
			\IF {(\ref{sd2}) is valid}
				\STATE Set $k=k-1$ and update $\tilde{D}_{k-1}^2$.
			\ENDIF
			\ENDIF
		\ENDIF 
  
		\IF{$k=0$ and $\tilde{D}_0^2<D_{\mathrm{min}}^2$}
			\STATE Set $D_{\mathrm{min}}^2=\tilde{D}_0^2$ and save $(\bm{s},\bm{v})$. 
		\ENDIF 

		\ENDWHILE
		\STATE \textbf{Output} $(\bm{s},\bm{v})$.
    \end{algorithmic}
\end{algorithm}
\end{spacing}

\section{Simulation Results}

In this section, we provide simulation results to verify the effectiveness of the proposed SRPM scheme. For a general case, we model the static channel between the BS and RIS, $\bm{G}$, as
\begin{align}
\bm{G}=\sqrt{\frac{NN_t}{N_p}}\left(\sum_{i=1}^{N_p} \beta_i \textbf{a}_N(\phi_{r,i}) \textbf{a}_{N_t}^H(\phi_{t,i})\right)
\end{align}
where $\beta_i$ is the complex gain of the $i$th path, $\phi_{r,i}$ and $\phi_{t,i}$ respectively represent the AoA of the $i$th path at the RIS and the AoD of the $i$th path at the BS, $N_p$ is the number of paths. Without loss of generality, the large-scale path losses are normalized, which does not affect the analysis of the results \cite{rsm}. Unless otherwise specified, the parameters are set as: the number of reflecting elements, $N=128$, the number of transmit antennas at the BS, $N_t=8$, the number of receive antennas, $N_r=4$, the number of sub-surfaces, $L=2$, the number of data streams, $N_s=1$, the modulation order of the message transmitted by the BS, $M=2$, and the antenna separation space is half of the wavelength, i.e., $\frac{d}{\lambda}=\frac{1}{2}$. Moreover, rectangular QAM is used. For comparison, we mainly considered the following benchmarks.
\begin{itemize}
\item RIS-PBF \cite{wu}: RIS is only used for passive beamforming to enhance the signal strength at the receiver and does not convey any extra message.
\item PBIT \cite{pbit}: For each sub-surface, it is randomly turned on and off with equal probability to convey one extra bit.
\item RIS-RPM \cite{rpm}: $p$ sub-surfaces among the $L$ sub-surfaces are turned off, where $p<L$ is a given number.
\item RIS-QRM \cite{qrm}: The phase offsets of the $p$ sub-surfaces are set as $\frac{\pi}{2}$, while the phase offsets of the other $(L-p)$ sub-surfaces are set as $0$.
\item The proposed SRPM: For each sub-surface, the same tunable phase offset is superimposed on the base phases for extra information transfer.
\end{itemize}

\subsection{Comparison Between the Numerical and Analytical Results}
\begin{figure}[!t]
\centering
	\begin{minipage}{0.4\linewidth}
		\centering
		\includegraphics[width=1\linewidth]{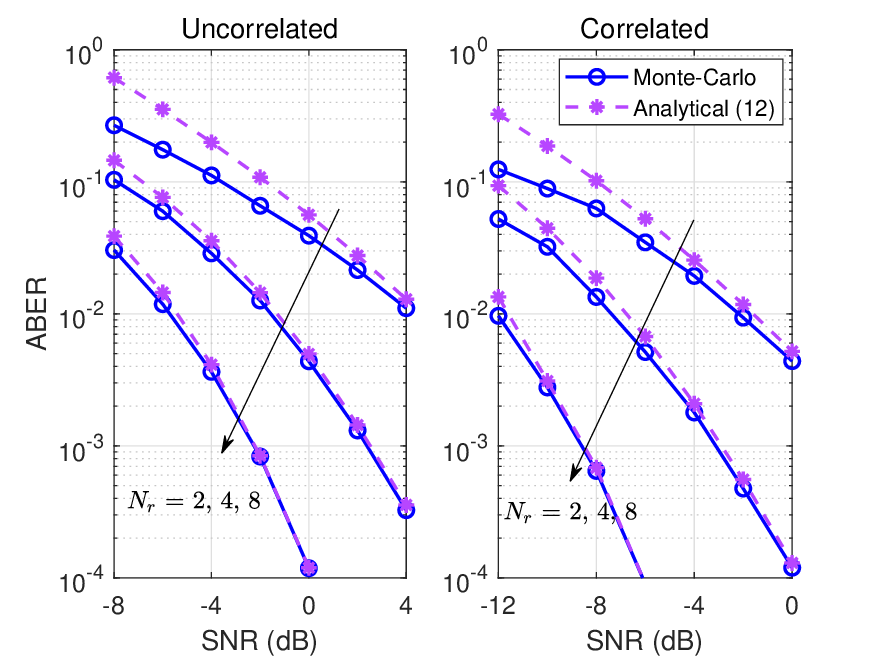}
	\end{minipage}
	\begin{minipage}{0.4\linewidth}
		\centering
		\includegraphics[width=1\linewidth]{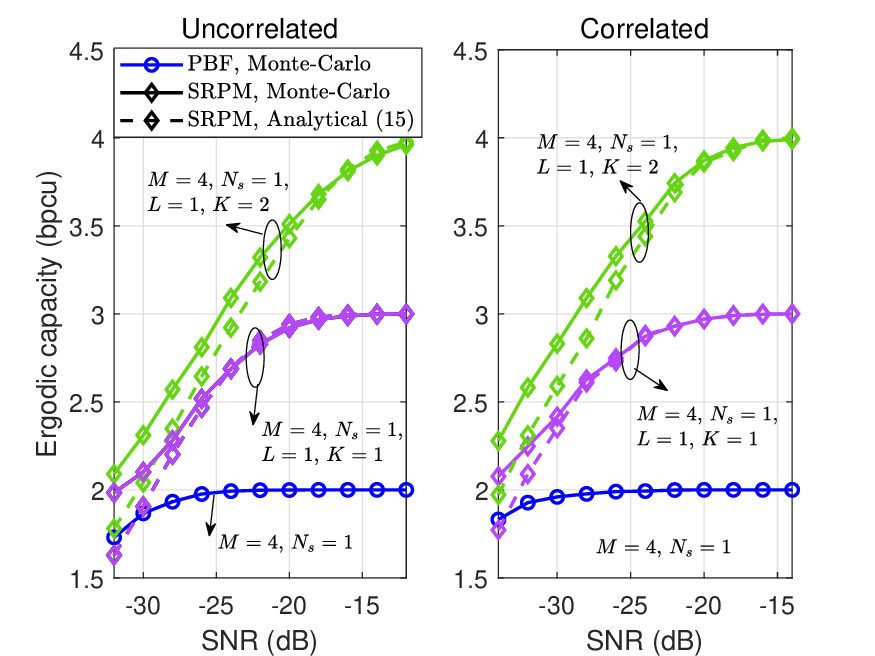}
	\end{minipage}
\caption{Numerical and analytical performance of the SRPM. (a) ABER. (b) Ergodic capacity.}
\label{fig3}
\vspace{-0.7cm}
\end{figure}
In Fig. \ref{fig3}, we compare the  numerical ABER and ergodic capacity obtained from Monte Carlo simulations and the analytical ABER derived in (\ref{eq17}) and (\ref{capacity}), respectively. It is observed that for all tested setups the analytical ABER and ergodic capacity asymptotically tight with growing SNR. Moreover, we observe an increase in diversity order if the number of receive antennas, $N_r$, increases, which also confirms \emph{Theorem 1}. The more receive antennas are deployed, the better ABER performance it achieves. As shown in Fig. \ref{fig3}(b), it is found that the proposed SRPM exhibits much higher channel capacity at high SNRs as expected. At low SNRs, the proposed SRPM also outperforms the traditional PBF scheme despite the performance gain becomes less significant. Both theoretical and simulation results validate the noticeable advantages of the proposed SRPM from the perspective of channel capacity.

\subsection{Impact of SRPM Parameters}
\begin{figure}[!t]
\centering
	\begin{minipage}{0.4\linewidth}
		\centering 
		\includegraphics[width=1\linewidth]{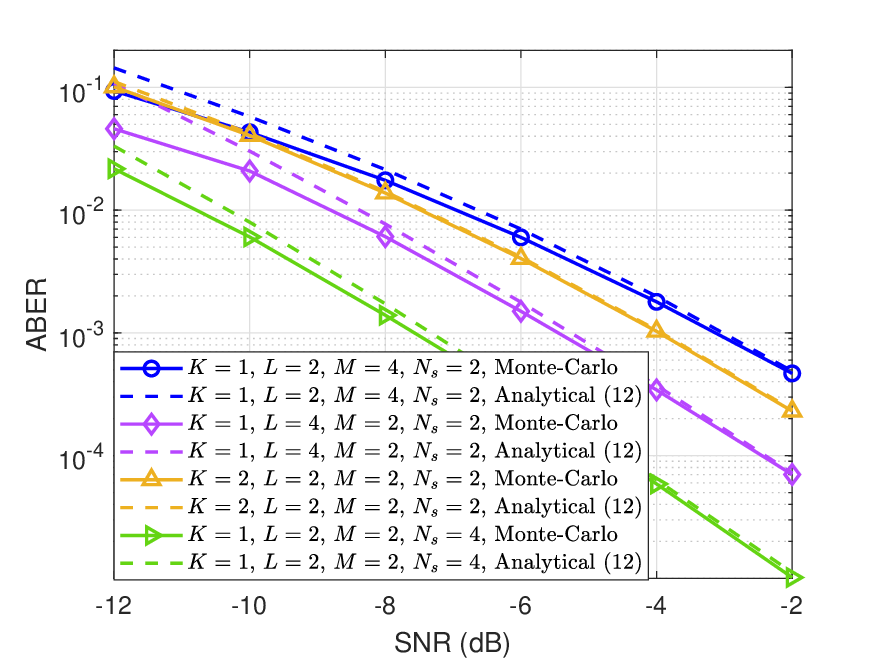}
	\end{minipage}
	\begin{minipage}{0.4\linewidth}
		\centering
		\includegraphics[width=1\linewidth]{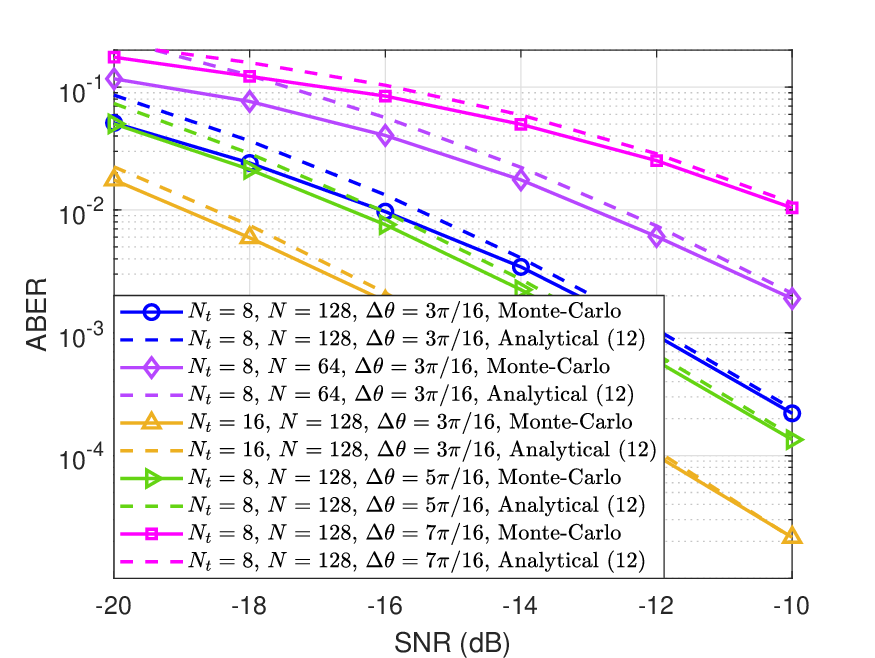}
	\end{minipage}
\caption{(a) ABER of the different modulation parameters. (b) ABER of the different system parameters.}
\vspace{-0.75cm}
\end{figure}
We evaluate the impact of various parameters, including the modulation parameters $(K,L,M,$ $N_s)$ and system parameters $(N_t, N,\Delta \theta)$. In Fig. 4(a), the ABER of different modulation parameter settings are compared for the same spectral efficiency of 4 bpcu. It is seen that the SRPM with parameters $K=1$, $L=2$, $M=2$, and $N_s=4$ outperforms that with all the other parameter settings. By contrast, the SRPM with parameters $K=1$, $L=2$, $M=4$, and $N_s=2$ performs the worst, which utilizes a higher-order modulation at the BS. Since lower modulation order correspond to a larger distance between adjacent constellation symbols, a lower detection error rate is achieved by increasing the number of symbols with lower modulation order. Similarly, we also observe that the SRPM with lager $L$ and smaller $K$ is superior to that with larger $K$ and smaller $L$. This is because increasing the modulation order, $K$, leads to closer constellation points, which come with more chances of error. Hence, it is preferred to increase the number of data streams or sub-surfaces rather than modulation order to pursue higher communication rate.

In Fig. 4(b), we evaluate the impact of the system parameters, which do not affect the communication rate, on the ABER. The numerical results verify that a 3 dB gain in terms of ABER is obtained by doubling either the number of transmit antennas, $N_r$, or the number of  reflecting elements, $N$. This is because the reflecting elements at the RIS play a similar role as the transmit antennas at BS from the perspective of achieving higher diversity gain. It suggests that we can achieve lower ABER with low energy consumption by deploying more low-cost reflecting elements at the RIS. With the growth of the unit step of the phase offsets, $\Delta \theta$, we observe that the ABER decrease first and then increase. This is because tiny phase offset results in small distances between the new and original constellation symbols. On the other hand, excessive phase rotation causes the new and adjacent constellation symbols to be too close. Apparently, there should exist an optimal phase offset for given $K$ and $M$, which can be obtained by exhaustive search for these discrete phase shifts.

\vspace{-0.4cm}
\subsection{Comparison Between Different Modulation Schemes}
\begin{figure}[!t]
\centering
	\begin{minipage}{0.4\linewidth}
		\centering
		\includegraphics[width=1\linewidth]{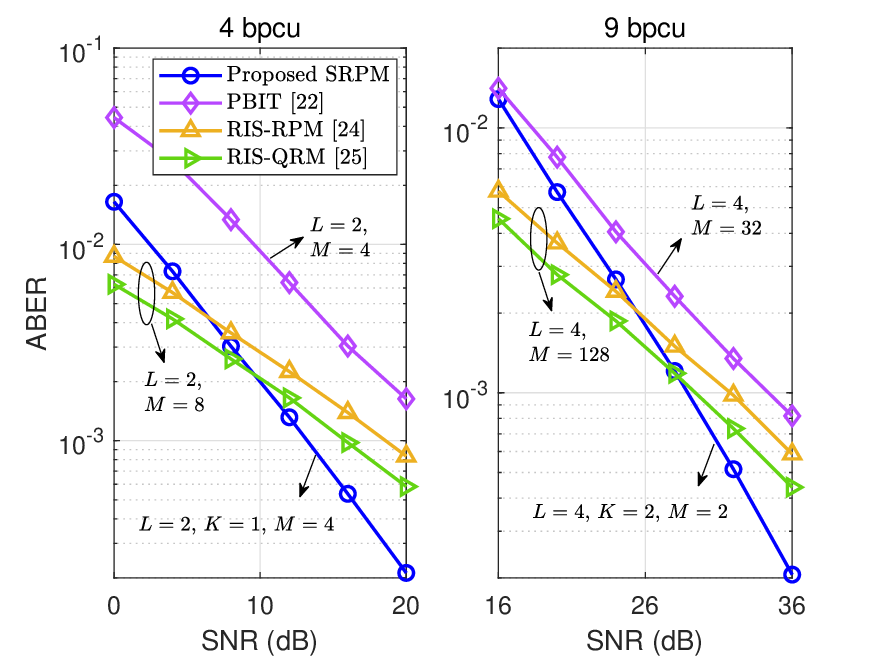}
	\caption{ABER of different schemes with single-antenna receiver.}
	\end{minipage}
	\begin{minipage}{0.4\linewidth}
		\centering
		\includegraphics[width=1\linewidth]{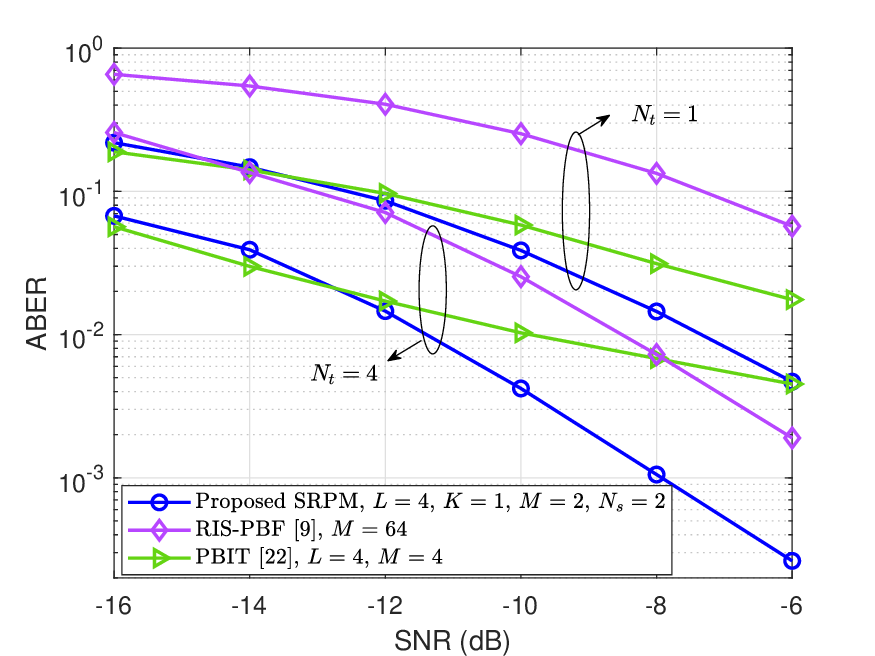}
\caption{ABER of different schemes with multi-antenna receiver.}
	\end{minipage}
\vspace{-0.75cm}
\end{figure}
Then, we consider the performance comparison between the SRPM scheme and state-of-the-art schemes. Firstly, we consider the case where a single antenna is equipped at the receiver, i.e., $N_r=1$, in Fig. 5. For the baselines, the high-order modulation at RIS is not supported and hence the value of $K$ is equal to one, which is omitted for simplicity. The proposed scheme~outperforms all the benchmarks as the SNR increases although the RIS-RPM and RIS-QRM schemes exhibit some performance advantages at the low SNRs. The proposed SRPM scheme is validated with a higher diversity order and the performance gap gradually increases with the SNR.

%
Fig. 6 depicts the performance comparison under the multi-antenna receiver case. For PBIT scheme, we see a high outage probability since a part of the reflecting elements are turned off. Compared with the RIS-PBF scheme, the proposed SRPM achieves a gain of 5 dB and 3 dB in terms of SNR respectively for the case that $N_t=1$ and $N_t=4$.  This phenomenon also confirms that PBF is a sub-optimal way of using RIS. Especially when there are few RF chains at the transmitter, the SRPM brings higher multiplexing gain by modulating information at RIS.

\vspace{-0.3cm}
\subsection{Comparison Between Different Precodings}

\begin{figure}[!t]
\centering
	\begin{minipage}{0.4\linewidth}
		\centering
		\includegraphics[width=1\linewidth]{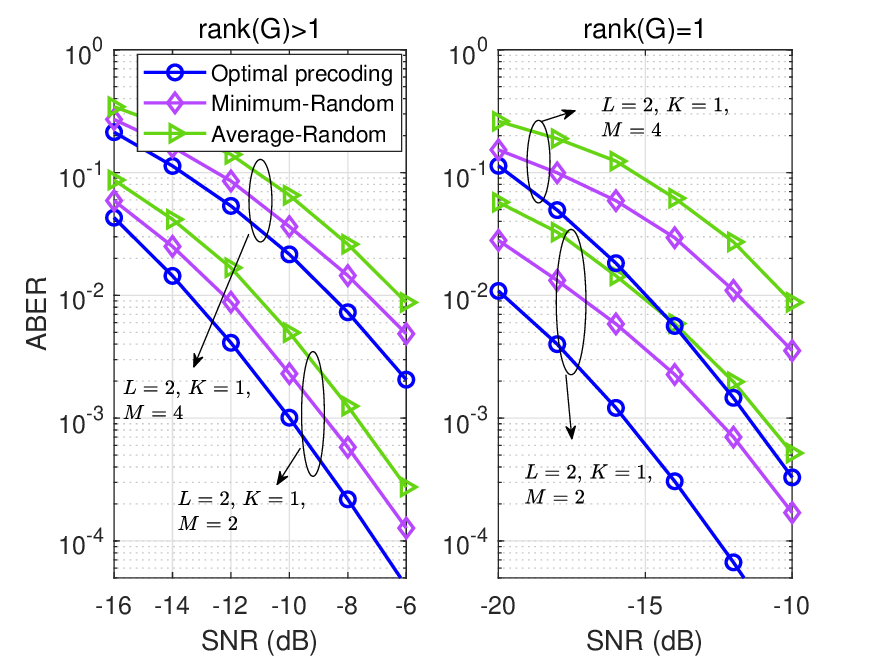}
	\caption{ABER of different precodings.}
	\end{minipage}
	\begin{minipage}{0.4\linewidth}
		\centering
		\includegraphics[width=1\linewidth]{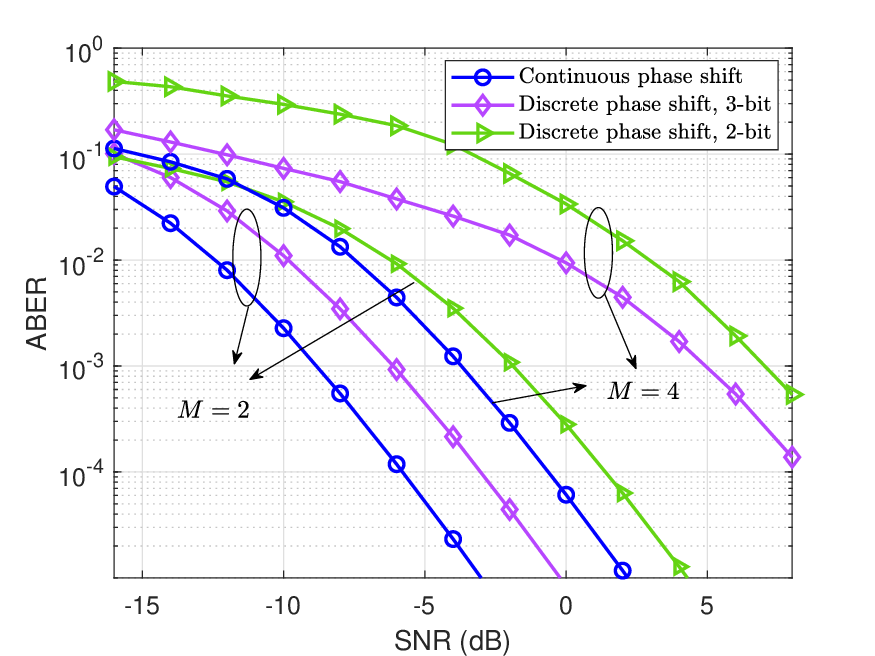}
\caption{ABER under discrete phase~shifts.}
	\end{minipage}
\vspace{-0.6cm}
\end{figure}
In Fig. 7, we evaluate the performance of the proposed precoding optimization in Section \Rmnum{4}.
The baselines with legend ``Minimum-Random" and ``Average-Random" represent minimum and average ABER obtained from more than 30 randomly selected precodings, respectively. As shown in Fig. 7, the optimized precoding brings significant performance gain by the SRPM.  With high rank $\bm{G}$, the performance gain in terms of SNR is about 2 dB. For rank-one $\bm{G}$, the obtained performance improvement grows up to about 5 dB. This is because for the rank-one $\bm{G}$, there only exists one major LoS path between the BS and RIS. On the contrary, due to the existence of multiple LoS paths for high rank $\bm{G}$, the performance is less sensitive to the direction of the beam, thereby resulting in the limited performance gain brought by optimization.
\vspace{-0.3cm}

\subsection{Impact of Discrete Phase Shifts}
In Fig. 8, the effectiveness of the SRPM with discrete phase shifts at RIS is validated. From this figure, the diversity order of the proposed SRPM remains unchanged with the discrete phase shift control when SNR grows large. In addition, the discrete phase shift causes a significant performance loss in terms of SNR, especially for higher order modulation at the BS.

\vspace{-0.5cm}
\subsection{Comparison Between the ML and SD-based Layered Detection}

\begin{figure}[!t]
\centering
	\begin{minipage}{0.4\linewidth}
		\centering
		\includegraphics[width=1\linewidth]{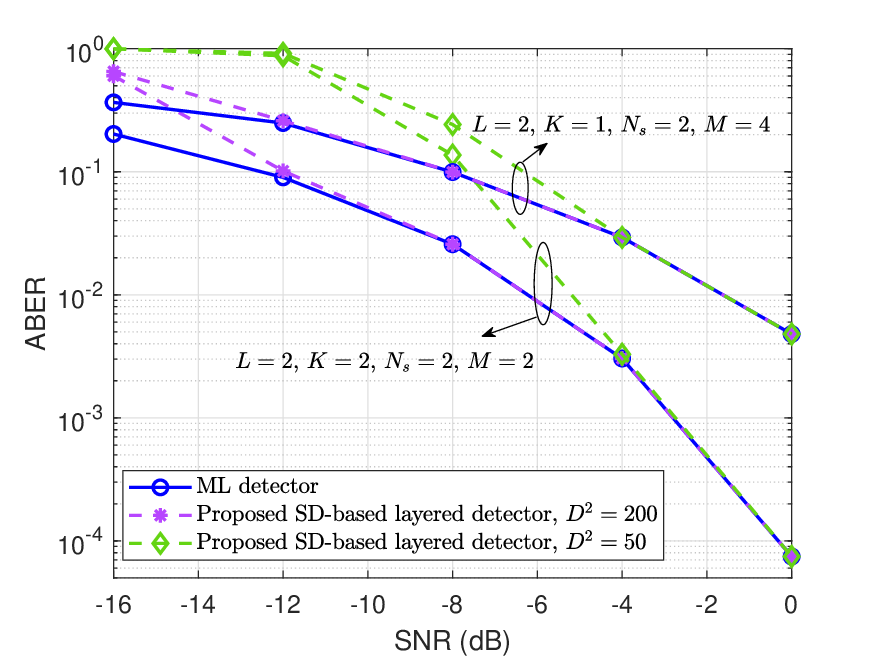}
	\end{minipage}
	\begin{minipage}{0.4\linewidth}
		\centering
		\includegraphics[width=1\linewidth]{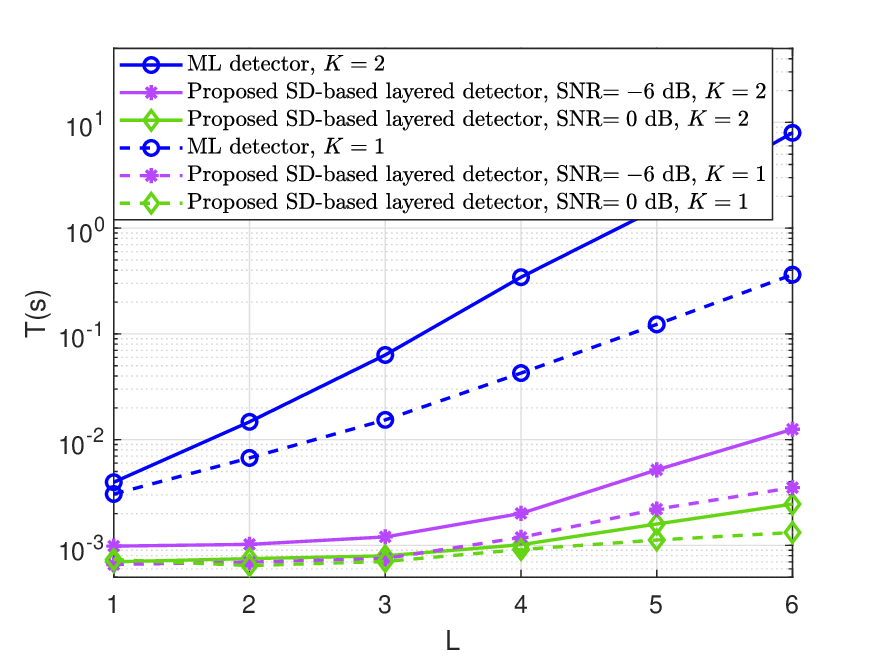}
	\end{minipage}
\caption{ (a) ABER of the ML and SD-based layered detection. (b) Running time of the ML and SD-based layered detection.}
\vspace{-0.6cm}
\end{figure}

Finally, we verify the effectiveness of the proposed SD-based layered detector. As shown in Fig. 9(a), we observe that the proposed SD-based layered detector asymptotically achieves the performance close to that of the ML detector as the increase of SNR. In particular, with sufficient large SNR, the proposed SD-based layered detection can achieve the optimal performance even with a small search radius $D$.

In Fig. 9(b), we compare the running time of the ML and SD-based layered detection, which reflects the complexity. Note that for the ML detection, the complexity grows exponentially with the number of sub-surfaces, $L$, and the number of data streams, $N_s$. It is seen from Fig. 9 that, compared with the ML detection, the reduction of running time brought by the proposed SD-based layered detection becomes quite pronounced upon increasing $L$. The value of complexity reduction parameter $\xi$ is approximately 0.27 and 0.13 when SNR is equal to -6 and 0 dB, respectively, which further demonstrates the low-complexity characteristic of the proposed scheme, especially under high SNR regimes.

\vspace{-0.4cm}
\section{Conclusion}
In this paper, we propose a novel paradigm of information transfer for the RIS-assisted MIMO systems to achieve higher communication rate. By superimposing information-bearing phase offsets onto predetermined base phases at RIS reflecting elements, extra messages can be transferred without the deployment of other RF chains. We derive the analytical ABER and ergodic capacity with the ML detection and further confirm that the SRPM achieves the optimal diversity order of $N_r$ for arbitrary parameters. Next, based on the derived ABER, we formulate a general precoding optimization framework to minimize the ABER and maximize the ergodic capacity, and confirm that the optimal solutions are available by exploiting the SDR technique. Moreover, we propose a two-layer framework to extend the traditional SD scheme to the SRPM detection. Numerical results demonstrate that the proposed SRPM outperforms the existing scheme in terms of ABER and ergodic capacity. Also, it is shown that we can effectively improve the spectral efficiency by increasing the number of data streams transmitted by BS and the number of sub-surfaces at RIS. The extension of SRPM to multiuser cases and scenarios with high mobility remains an interesting and challenging research direction.  In addition, how to fully utilize the DoF brought by the adjustable amplitude in active RIS to achieve higher-order modulation is also an interesting topic worth considering.

\appendices
\section{Proof of Lemma 1}
Firstly, according to (\ref{eq5}), we express the $i$th element of the received signal $\bm{y}$ in detail as
\begin{align}
y_i&=\sum_{m=1}^{N_s}s_m \sum_{p=1}^{N_t} w_{p,m} h_{d,i,p}+\sum_{l=1}^L \sum_{m=1}^{N_s} s_m e^{j k_l \Delta \theta} \sum_{n\in \mathcal{A}_l} g_{n,m} e^{j\theta_n^b} h_{r,i,n} +z_i\nonumber \\
&=\sum_{p=1}^{N_t}h_{d,i,p} \sum_{m=1}^{N_s} w_{p,m}s_m +\sum_{l=1}^L \sum_{n\in \mathcal{A}_l}h_{r,i,n}
\sum_{m=1}^{N_s} g_{n,m} e^{j\theta_n^b} s_m e^{j k_l \Delta \theta}    +z_i,
\end{align}
where $w_{p,m}$ is the $(p,m)$th element of $\bm{W}$, $h_{d,i,p}$ is the $(i,p)$th element of $\bm{H}_d$, and $h_{r,i,n}$ is the $(i,n)$th element of $\bm{H}_r$. We denote the $(n,m)$th element of $\bm{GW}$ by $g_{n,m}$, which is a constant. Then, we rewrite the  $i$th element of $\bm{\delta}$ in (\ref{delta}) as
\begin{align}\label{deltai}
\delta_i&=\sqrt{\beta_d}\sum_{p=1}^{N_t}h_{d,i,p} \sum_{m=1}^{N_s} w_{p,m}(s_m-\hat{s}_m) +\sqrt{\beta_r}\sum_{l=1}^L \sum_{n\in \mathcal{A}_l}h_{r,i,n} \sum_{m=1}^{N_s} g_{n,m} e^{j\theta_n^b} \left(s_m e^{j k_l \Delta \theta}-\hat{s}_m e^{j \hat{k}_l \Delta \theta}\right)\nonumber \\
&=\sqrt{\beta_d}\sum_{p=1}^{N_t} \sum_{\imath=1}^{N_r} \sum_{\jmath=1}^{N_t} h_{\omega,\imath,\jmath}^d  r_{u,i,\imath} r_{b,\jmath,p} \sum_{m=1}^{N_s} w_{p,m}(s_m-\hat{s}_m) \nonumber \\
&\quad +\sqrt{\beta_r}\sum_{l=1}^L \sum_{n\in \mathcal{A}_l} \sum_{\imath=1}^{N_r} \sum_{\jmath=1}^{N} h_{\omega,\imath,\jmath}^r  r_{u,i,\imath} r_{r,\jmath,n} \sum_{m=1}^{N_s} g_{n,m} e^{j\theta_n^b} \left(s_m e^{j k_l \Delta \theta}-\hat{s}_m e^{j \hat{k}_l \Delta \theta}\right),
\end{align}
where  $\hat{s}_m$ and $\hat{k}_l$ are the detected version of $s_m$ and $k_l$, respectively, $h_{\omega,\imath,\jmath}^d$ is the $(\imath,\jmath)$-th element of $\bm{H}_{\omega}^d$, $h_{\omega,\imath,\jmath}^d$ is the $(\imath,\jmath)$-th element of $\bm{H}_{\omega}^r$, $r_{u,i,\imath}$ is the $(i,\imath)$-th element of $\bm{R}_{u}$, $r_{b,\jmath,p}$ is the $(\jmath,p)$-th element of $\bm{R}_{b}$, and $r_{r,\jmath,n}$ is the $(\jmath,n)$-th element of $\bm{R}_{r}$. Considering that $h_{\omega,\imath,\jmath}^d$ and $h_{\omega,\imath,\jmath}^r$ are i.i.d. complex Gaussian variables with zero mean, we conclude that $\delta_i$ is also a complex Gaussian random variable with zero mean. Let us denote the covariance matrix of $\bm{\delta}$ by $\bm{C}$. The $(i,j)$-th element of $\bm{C}$ is evaluated by
\begin{align}
\mathbb{E}\left\{ \delta_i \delta_j^* \right \}&=\beta_d \sum_{\imath=1}^{N_r} \sum_{\jmath=1}^{N_t}
 r_{u,i,\imath}r_{u,j,\imath}^*
\left \vert \sum_{p=1}^{N_t}\sum_{m=1}^{N_s} r_{b,\jmath,p} w_{p,m}(s_m-\hat{s}_m)\right \vert ^2
\nonumber \\
&\quad +  \beta_r \sum_{\imath=1}^{N_r} \sum_{\jmath=1}^{N} r_{u,i,\imath}  r_{u,j,\imath}^*
\left \vert \sum_{l=1}^L \sum_{n\in \mathcal{A}_l}\sum_{m=1}^{N_s} r_{r,\jmath,n} g_{n,m} e^{j\theta_n^b} \left(s_m e^{j k_l \Delta \theta}-\hat{s}_m e^{j \hat{k}_l \Delta \theta}\right) \right \vert^2.
\end{align}
Considering that $\lambda\triangleq \Vert \bm{\delta}\Vert^2= \bm{\delta}^H \bm{I}_{N_r} \bm{\delta}$, the MGF of $\lambda$ can thus be calculated by using \cite[Eq.~(7)]{chisquared} as
\begin{align}
\mathcal{M}_\lambda(t)= \left ( \mathrm{det} \left(\bm{I}_{N_r}-t\bm{C} \right) \right )^{-1}=\prod_{i=1}^{N_r} (1-\lambda_i t)^{-1},
\end{align}
where $\lambda_i$ is the $i$-th non-ordered eigenvalue of $\bm{C}$.

Furthermore, considering the ideal uncorrelated cases, $h_{d,i,p}$, $\forall i,p$, and $h_{r,i,n}$, $\forall i,n$ in (\ref{deltai}) are i.i.d. complex Gaussian random variables with zero mean and variance $\beta_d$ and $\beta_r$, respectively. We further define $c_{p,m}\triangleq \sum_{m=1}^{N_s}w_{p,m}(s_m-\hat{s}_m)$ and conclude that the first term in (\ref{deltai}) follows $\mathcal{CN} \left(0,\beta_d \sum_{p=1}^{N_t} \vert c_{p,m} \vert^2 \right)$.
Similarly, we also obtain that the second term in (\ref{deltai}) follows $\mathcal{CN}\left(0,\right.$ $\left. \beta_r \sum_{l=1}^L \sum_{n\in \mathcal{A}_l } \left \vert d_{l,n} \right \vert^2\right)$,
where $ d_{l,n}\triangleq \sum_{m=1}^{N_s} g_{n,m} \left(s_m e^{j k_l \Delta \theta}-\hat{s}_m e^{j \hat{k}_l \Delta \theta}\right)$. Therefore from (\ref{deltai}), we have
\begin{align}\label{sigmad}
\bm{\delta}\sim \mathcal{CN} \left( \bm{0},\sigma_d^2\bm{I}_{N_r}\right),
\end{align}
where $\sigma_d^2\triangleq \beta_d \sum_{p=1}^{N_t} \vert c_{p,m} \vert^2+ \beta_r \sum_{l=1}^L \sum_{n\in \mathcal{A}_l } \left \vert d_{l,n} \right \vert^2$. Accordingly, it is found that $\lambda=\Vert \bm{\delta}\Vert^2$ is a Chi-squared random variable with the degrees of freedom $2N_r$. Hence, its MGF is also obtained by exploiting \cite[Eq. (7)]{chisquared} as
\begin{align}
\mathcal{M}_\lambda(t)=(1-\sigma_d^2t)^{-N_r}.
\end{align}
The proof completes.

\section{Proof of Theorem 2}
Firstly, the mutual information is reformulated through the chain rule as
\begin{align}\label{aeq2}
\mathcal{I}(\bm{s},\bm{v};\bm{y})=\mathcal{H}(\bm{y}) -\mathcal{H} (\bm{y} \vert \bm{s},\bm{v}),
\end{align}
where $\mathcal{H}(\cdot)$ denotes the differential entropy. Firstly, from the system model in (\ref{eq3}), the conditional differential entropy $\mathcal{H} (\bm{y} \vert \bm{s},\bm{v})$ in (\ref{aeq2}) is written as
\begin{align}\label{aeq3}
\mathcal{H} (\bm{y} \vert \bm{s},\bm{v})&=\mathcal{H} \left(\sqrt{P} (\bm{H}_d +\bm{H}_r \bm{\Theta}^b \bm{\Xi} \bm{G})\bm{W}\bm{s} +\bm{z} \vert \bm{s},\bm{v}\right)=\mathcal{H}(\bm{z})=N_r \log_2(\pi e \sigma^2),
\end{align}
where $\bm{z}\sim \mathcal{CN}(\bm{0},\sigma^2 \bm{I}_{N_r})$ and the second equality holds because $\sqrt{P} (\bm{H}_d +\bm{H}_r \bm{\Theta}^b \bm{\Xi} \bm{G})\bm{W}\bm{s}$ is a constant given $(\bm{s},\bm{v})$, $\bm{H}_d$, $\bm{H}_r$, and $\bm{G}$. As for the term $\mathcal{H}(\bm{y})$ in (\ref{aeq2}), we have
\begin{align}\label{aeq4}
\mathcal{H}(\bm{y})= -\int_{\mathbb{C}} f_{\bm{y}}(\bm{y})\log_2 \left( f_{\bm{y}}(\bm{y})\right) \mathrm{d}\bm{y},
\end{align}
where $f_{\bm{y}}(\cdot)$ is the PDF of $\bm{y}$.
By exploiting the observation that the full capacity is achieved by using equiprobable inputs \cite[Eq. (9)]{capacity1}, $f_{\bm{y}}(\bm{y})$ is further expressed as
\begin{align}\label{aeq5}
f_{\bm{y}}(\bm{y})&=\sum_{\bm{s}} \sum_{\bm{v}} \frac{1}{M^{N_s}(2K+1)^L} f(\bm{y}\vert \bm{s},\bm{v}) \nonumber \\
&= \sum_{\bm{s}} \sum_{\bm{v}} \frac{1}{S (\pi \sigma^2)^{N_r}} \exp \left ( -\frac{1}{\sigma^2}\left \Vert \bm{y}- \sqrt{P}(\bm{H}_d +\bm{H}_r \bm{\Theta}^b \bm{\Xi} \bm{G})\bm{W}\bm{s} \right \Vert^2 \right ),
\end{align} 
where $S\triangleq M^{N_s}(2K+1)^L$.
Plugging (\ref{aeq5}) into (\ref{aeq4}) and using the Jensen's inequality, we have
\begin{align}\label{aeq6}
\mathcal{H}(\bm{y})&\geq -\log_2 \left( \int_{\mathbb{C}} f_{\bm{y}}(\bm{y})^2 \mathrm{d}\bm{y}\right)\nonumber \\
&= -\log_2 \left(\frac{1}{S^2 (\pi\sigma^2)^{2N_r}}\int_{\mathbb{C}}  \sum_{\bm{s}} \sum_{\bm{v}}\sum_{\hat{\bm{s}}} \sum_{\hat{\bm{v}}}\exp\left ( -\frac{\left \Vert\bm{z} \right \Vert^2+\left \Vert \bm{z}+ \sqrt{P}\bm{\delta} \right \Vert^2}{\sigma^2} \right )\mathrm{d}\bm{z}\right)\nonumber \\
&\overset{(a)}{=}-\log_2 \left(\frac{1}{S^2 (2\pi\sigma^2)^{N_r}}\sum_{\bm{s}} \sum_{\bm{v}}\sum_{\hat{\bm{s}}} \sum_{\hat{\bm{v}}} \exp \left(-\frac{P\Vert \bm{\delta}\Vert^2 }{2\sigma^2} \right)\right)\nonumber \\
&\overset{(b)}{=}2 \log_2 S+ N_r \log_2\left(2\pi\sigma^2 \right)-\log_2 \left(S+\sum_{\bm{s}} \sum_{\bm{v}}\sum_{\hat{\bm{s}}\neq \bm{s}} \sum_{\hat{\bm{v}}\neq \bm{v}} \exp \left(-\frac{P\Vert \bm{\delta}\Vert^2 }{2\sigma^2} \right)\right),
\end{align}
where $(a)$ exploits the property of the Gaussian PDF and $(b)$ is due to the fact that $\bm{\delta}=\bm{0}$ when $(\hat{\bm{s}},\hat{\bm{v}})=(\bm{s},\bm{v})$.
Hence, by substituting (\ref{aeq2}), (\ref{aeq3}), and (\ref{aeq6}) into (1) and exploiting the approximation in \cite[Eq. (27)]{capacity2}, we have
\begin{align}
C_{\mathrm{SRPM}}&\approx 2 \log_2 S-\log_2 \left(S+\sum_{\bm{s}} \sum_{\bm{v}}\sum_{\hat{\bm{s}}\neq \bm{s}} \sum_{\hat{\bm{v}}\neq \bm{v}} \mathbb{E}\left\{\exp \left(-\frac{P\Vert \bm{\delta}\Vert^2 }{2\sigma^2} \right)\right \}\right)\nonumber \\
&=2\log_2 S-\log_2 \left(S+\sum_{\bm{s}} \sum_{\bm{v}}\sum_{\hat{\bm{s}}\neq \bm{s}} \sum_{\hat{\bm{v}}\neq \bm{v}} \mathcal{M}_{\lambda}\left(-\frac{P}{2\sigma^2}\right)\right),
\end{align}
which completes the proof.

\section{Proof of Lemma 2}
For the single-stream case, the $i$th element of $\bm{\delta}$ in (\ref{deltai}) becomes
\begin{align}\label{eeq37}
\delta_i=\sum_{p=1}^{N_t}h_{d,i,p} w_{p}(s-\hat{s}) +\sum_{l=1}^L \sum_{n\in \mathcal{A}_l}h_{r,i,n}
 e^{j\theta_n^b} \left(s e^{j k_l \Delta \theta}-\hat{s} e^{j \hat{k}_l \Delta \theta}\right)\bm{g}_n^H \bm{w},
\end{align}
where $w_p$ is the $p$th element of vertor $\bm{w}$, and $\bm{g}_n^H$ represents the $n$th row of $\bm{G}$. Analogously, we conclude that the first term of $\delta_i$, i.e., $\sum_{p=1}^{N_t}h_{d,i,p} w_{p}(s-\hat{s})$, follows a Gaussian distribution with zero mean. Its variance is calculated as
\begin{align}
\beta_d \sum_{p=1}^{N_t} \left \vert w_p (s-\hat{s})\right \vert^2=\beta_d \vert s-\hat{s}\vert^2 \sum_{p=1}^{N_t} \left \vert w_p \right \vert^2=\beta_d \vert s-\hat{s}\vert^2 \Vert \bm{w} \Vert^2=\beta_d \vert s-\hat{s}\vert^2,
\end{align}
which is irrespective with the precoding vector. Also, the variance of the second summation term in (\ref{eeq37}) is formulated as
\begin{align}
 \beta_r \sum_{l=1}^L \sum_{n\in \mathcal{A}_l}
\left \vert  \left(s e^{j k_l \Delta \theta}-\hat{s} e^{j \hat{k}_l \Delta \theta}\right)\bm{g}_n^H \bm{w}\right \vert^2&= \beta_r\sum_{l=1}^L \left \vert  \left(s e^{j k_l \Delta \theta}-\hat{s} e^{j \hat{k}_l \Delta \theta}\right) \right \vert^2 \sum_{n\in \mathcal{A}_l} \bm{w}^H \bm{g}_n \bm{g}_n^H \bm{w}\nonumber \\
=&\bm{w}^H  \bm{A}_{(s,\bm{v})\to(\hat{s},\hat{\bm{v}})} \bm{w},
\end{align}
where $\bm{A}_{(s,\bm{v})\to(\hat{s},\hat{\bm{v}})} \triangleq  \beta_r \sum_{l=1}^L \left \vert  \left(s e^{j k_l \Delta \theta}-\hat{s} e^{j \hat{k}_l \Delta \theta}\right) \right \vert^2 \sum_{n\in \mathcal{A}_l}  \bm{g}_n \bm{g}_n^H$ is  positive semidefinite. Hence, the $\sigma_d^2$ defined in (\ref{sigmad}) becomes
\begin{align}
\sigma_d^2 =\beta_d  \vert s-\hat{s}\vert^2+\bm{w}^H  \bm{A}_{(s,\bm{v})\to(\hat{s},\hat{\bm{v}})} \bm{w}.
\end{align}
Then, we rely on the accurate exponential bounds of the $Q$-function, that is
\begin{align}
\mathcal{Q}(x) \leq \sum_{i=1}^Q a_i e^{-b_i x^2},
\end{align}
where $a_i=\frac{\theta_i-\theta_{i-1}}{\pi}$, $b_i=\frac{1}{2\sin^2\theta_i}$, and $0=\theta_0\leq \theta_1\leq\cdots\leq \theta_Q=\pi/2$ \cite[Eq. (8)]{approx}. Combining the results in \emph{Lemma 1}, we obtain the desired ABER bound in (\ref{eq18}). The proof completes.

\section{Proof of Theorem 3}
To start, we formulate the Lagrangian function of the problem in (\ref{problem3}) as
\begin{align}
\mathcal{L}(\bm{W},\nu,\bm{\Lambda})=f(\bm{W})+\nu\left(\mathrm{Tr}(\bm{W})-1\right)-\mathrm{Tr}(\bm{\Lambda}\bm{W}),
\end{align}
where $\nu$ and $\bm{\Lambda}$ are the Lagrange multipliers for the constraints $\mathrm{Tr}(\bm{W})=1$ and $\bm{W}\succeq\bm{0}$, respectively. Considering that the  problem in (\ref{problem3}) is a convex problem, the  Karush-Kuhn-Tucker (KKT) conditions are listed as follows
\begin{align}
&\mathrm{K}1: \enspace \nabla \mathcal{L}(\bm{W}^{\mathrm{opt}},\nu,\bm{\Lambda})=\bm{0}\iff \nu\bm{I}_{N_t}-\bm{\Gamma}=\bm{\Lambda}, \nonumber \\
&\mathrm{K}2: \enspace \mathrm{Tr}(\bm{W}^{\mathrm{opt}})=1,\enspace \bm{\Lambda}\bm{W}^{\mathrm{opt}}=\bm{0},\nonumber \\
&\mathrm{K}3: \enspace \bm{\Lambda}\succeq \bm{0}.
\end{align}
where $\bm{\Gamma} \triangleq \sum_{s} \sum_{\bm{v}} \sum_{\hat{s}} \sum_{\hat{\bm{v}}}\frac{e(s,\bm{v} \to \hat{s},\hat{\bm{v}})N_r}{\left( 1+\frac{P}{4\sigma^2}\left(\beta_d\vert s-\hat{s} \vert^2+\mathrm{Tr}\left(\bm{A}_{(s,\bm{v} )\to( \hat{s},\hat{\bm{v}})}\bm{W}^{\mathrm{opt}} \right)\right)\right)^{N_r+1}}\bm{A}_{(s,\bm{v} \to \hat{s},\hat{\bm{v}})}$. Also, we can conclude that $\bm{\Gamma}\succeq\bm{0}$.

Let us denote the maximum eigenvalue of $\bm{\Gamma}$ by $\lambda_{\mathrm{max}}$. If $\nu<\lambda_{\max}$, we have that $\bm{\Lambda}$ cannot be positive semidefinite, which is inconsistent with condition $\mathrm{K}3$. Otherwise, when $\nu>\lambda_{\max}$, we obtain that $\bm{\Lambda}$ is full rank. Combined with $\bm{\Lambda}\bm{W}^{\mathrm{opt}}=\bm{0}$, it is found that $\bm{W}^{\mathrm{opt}}=\bm{0}$, which contradicts $\mathrm{Tr}(\bm{W}^{\mathrm{opt}})=1$ in condition $\mathrm{K}2$. Hence, we finally arrive at $\nu =\lambda_{\max}$. Denote the normalized eigenvector of $\bm{\Gamma}$ associated with eigenvalue $\lambda_{\max}$ as $\bm{\xi}_{\max}$. Hence, there exists an optimal beamforming matrix $\bm{W}^{\mathrm{opt}}$ with rank one, i.e., $\bm{\xi}_{\max}\bm{\xi}_{\max}^H$, satisfying the KKT conditions $\mathrm{K}1-\mathrm{K}3$. The proof completes.
\vspace{-0.6cm}
\section{Proof of Lemma 3}
For $\bm{G}=\beta \bm{a}\bm{b}^H$, we define $\gamma\triangleq \bm{b}^H \bm{w}$. Then, the $i$th element of $\bm{\delta}$ in (\ref{deltai}) becomes
\begin{align}
\delta_i=\sum_{p=1}^{N_t}h_{d,i,p} w_{p}(s-\hat{s}) +\beta \gamma \sum_{l=1}^L \sum_{n\in \mathcal{A}_l}h_{r,i,n}
 e^{j\theta_n^b} \left(s e^{j k_l \Delta \theta}-\hat{s} e^{j \hat{k}_l \Delta \theta}\right)a_n ,
\end{align}
where $a_n$ is the $n$th element of $\bm{a}$. Similar to the derivations in Appendix C, we obtain that
\begin{align}
\sigma_d^2=\beta_d \vert s-\hat{s}\vert^2+\vert \gamma \vert^2 \beta^2\beta_d \frac{N}{L}\sum_{l=1}^L \left \vert s e^{jk_l\Delta \theta}-\hat{s}e^{j\hat{k}_l\Delta \theta}\right\vert^2.
\end{align}
Note that the design of precoding vector is decoupled from the specific pairwise error.  Hence, the problem in (\ref{problem3}) is equivalent to maximize $\left \vert \bm{b}^H \bm{w} \right \vert^2$ under the constraint $\Vert \bm{w} \Vert^2=1$.
The optimal solution is given by $\bm{w}^{\mathrm{opt}}=\frac{\bm{b}}{\Vert \bm{b}\Vert}$, which completes the proof.

\end{document}